\documentclass[10pt,conference, final]{IEEEtran}

\usepackage{amsmath,amssymb,amsthm}
\usepackage{xcolor}
\usepackage[linesnumbered,algoruled,vlined]{algorithm2e}
\usepackage{xparse} 
\usepackage{array,multirow}
\usepackage{graphicx}
\usepackage{enumitem}

\setlist[description]{%
  topsep=0pt,               
  itemsep=0pt,               
  labelwidth=0.0cm,
  labelindent=0pt,
  leftmargin=12pt
}

\usepackage{comment}
\usepackage{epstopdf}
\usepackage{balance}
\usepackage[hidelinks]{hyperref}
\usepackage[caption=false,font=footnotesize]{subfig}

\newcolumntype{F}{>{$\displaystyle}r<{$}@{\hspace{0.0em}}}
\newcolumntype{C}{>{$\displaystyle\,}c<{$}@{\hspace{0.0em}}}
\newcolumntype{B}{>{$\displaystyle\,}r<{$}@{\hspace{0.0em}}}
\newcolumntype{R}{>{$\displaystyle}r<{$}@{\hspace{0.2em}}}
\newcolumntype{S}{>{$\displaystyle}r<{$}@{\hspace{0.2em}}}
\newcolumntype{L}{>{$\displaystyle}l<{$}@{\hspace{0.2em}}}
\newcolumntype{Q}{>{$\displaystyle}l<{$}@{\hspace{0.3em}}}

 \newcounter{IPnumber}
 \newcommand{\tagIt}[1]{\refstepcounter{equation}\textnormal{({\theequation})} \label{#1}}
 
 \makeatletter
 
 \newcommand{\eqnum}{\leavevmode\hfill\refstepcounter{equation}\textup{\tagform@{\theequation}}}
 \makeatother

\usepackage{thmtools,thm-restate}
\usepackage{dblfloatfix}

\theoremstyle{plain}
\newtheorem{theorem}{Theorem}
\newtheorem{lemma}[theorem]{Lemma}

\newtheorem{definition}[theorem]{Definition}

\newtheorem{corollary}[theorem]{Corollary}

\makeatletter
\newcommand{\nosemic}{\renewcommand{\@endalgocfline}{\relax}}
\newcommand{\dosemic}{\renewcommand{\@endalgocfline}{\algocf@endline}}
\newcommand{\pushline}{\Indp}
\newcommand{\popline}{\Indm\dosemic}
\let\oldnl\nl
\newcommand{\nonl}{\renewcommand{\nl}{\let\nl\oldnl}}
\makeatother

\begin{document}
\IEEEoverridecommandlockouts
\IEEEpubid{\makebox[\columnwidth]{A short version of this work appeared in IFIP Networking 2018~\cite{rostSchmidVNEP_RR_IFIP_18}.~\hfill}\hspace{\columnsep}\makebox[\columnwidth]{}}

\title{Virtual Network Embedding Approximations: Leveraging Randomized Rounding}

\author{\IEEEauthorblockN{Matthias Rost}
\IEEEauthorblockA{TU Berlin, Germany\\
Email: mrost@inet.tu-berlin.de
}
\and
\IEEEauthorblockN{Stefan Schmid}
\IEEEauthorblockA{University of Vienna, Austria\\
Email: stefan\_schmid@univie.ac.at }
}

\maketitle

\newcommand{\TODO}[1]{\textcolor{red}{TODO: #1}}

\newcommand\numberthis{\addtocounter{equation}{1}\tag{\theequation}}

\newcommand{\NULL}{\textnormal{\texttt{NULL}}}
\newcommand{\scale}{\ensuremath{\lambda}}


\newcommand{\preals}{\ensuremath{\mathbb{R}_{\geq 0}}}


\newcommand{\requests}{\ensuremath{\mathcal{R}}}
\newcommand{\requestsP}{\ensuremath{\mathcal{R}'}}
\newcommand{\req}{r}
\newcommand{\types}{\ensuremath{\mathcal{T}}}
\newcommand{\type}{\ensuremath{\tau}}
\newcommand{\SVTypes}[1][\type]{\ensuremath{V^{#1}_{S}}}
\newcommand{\SVTypesCycle}[1][C_k]{\ensuremath{V^{#1}_{S,t}}}

\newcommand{\VG}[1][\req]{\ensuremath{G_{#1}}}
\newcommand{\VV}[1][\req]{\ensuremath{V_{#1}}}
\newcommand{\VE}[1][\req]{\ensuremath{E_{#1}}}
\newcommand{\VGbar}[1][\req]{\ensuremath{\bar{G}_{#1}}}
\newcommand{\VVbar}[1][\req]{\ensuremath{\bar{V}_{#1}}}
\newcommand{\VEbar}[1][\req]{\ensuremath{\bar{E}_{#1}}}
\newcommand{\Vstart}[1][\req]{\ensuremath{s_{#1}}}
\newcommand{\Vend}[1][\req]{\ensuremath{t_{#1}}}

\newcommand{\VGext}[1][\req]{\ensuremath{G^{\textnormal{ext}}_{#1}}}
\newcommand{\VGextFlow}[1][\req]{\ensuremath{G^{\textnormal{ext}}_{#1,f}}}
\NewDocumentCommand{\VGP}{O{\req} O{i} O{j}}{\ensuremath{G^{#2,#3}_{#1}}}
\newcommand{\VVext}[1][\req]{\ensuremath{V^{\textnormal{ext}}_{#1}}}
\newcommand{\VEext}[1][\req]{\ensuremath{E^{\textnormal{ext}}_{#1}}}
\newcommand{\VEextHorizontal}[1][\req]{\ensuremath{E^{\textnormal{ext}}_{#1,u,v}}}
\newcommand{\VEextVertical}[1][\req]{\ensuremath{E^{\textnormal{ext}}_{#1,\type,u}}}

\newcommand{\Vsource}[1][\req]{\ensuremath{o^+_{\req}}}
\newcommand{\Vsink}[1][\req]{\ensuremath{o^-_{\req}}}

\newcommand{\VMultiplicity}[1][\req]{\ensuremath{M_{#1}}}
\newcommand{\Vprofit}[1][\req]{\ensuremath{b_{#1}}}
\newcommand{\VprofitMax}{\ensuremath{b_{\max}}}

\newcommand{\Vcap}[1][\req]{\ensuremath{d_{#1}}}
\newcommand{\VVloc}[1][i]{\ensuremath{V^{{\req,#1}}_{S}}}
\newcommand{\VEloc}[1][i,j]{\ensuremath{E^{{\req,#1}}_{S}}}
\newcommand{\Vtype}[1][\req]{\ensuremath{\tau_{#1}}}


\newcommand{\SG}{\ensuremath{G_S}}
\newcommand{\SR}{\ensuremath{R_{S}}}
\newcommand{\SRV}{\ensuremath{R^V_{S}}}
\newcommand{\SV}{\ensuremath{V_S}}
\newcommand{\SE}{\ensuremath{E_S}}

\newcommand{\Scap}{\ensuremath{d_{S}}}
\newcommand{\ScapType}[1][\type]{\ensuremath{d^{#1}_{\SV}}}
\newcommand{\ScapTypePrime}[1][\type]{\ensuremath{{{d'}^{#1}_{\SV}}}}

\newcommand{\Scost}{\ensuremath{c_{S}}}
\newcommand{\ScostType}[1][\type]{\ensuremath{c^{#1}_{\SV}}}


\newcommand{\map}[1][\req]{\ensuremath{m_{#1}}}
\newcommand{\mapV}[1][\req]{\ensuremath{m^V_{#1}}}
\newcommand{\mapE}[1][\req]{\ensuremath{m^E_{#1}}}

\newcommand{\FeasibleLP}{\ensuremath{\mathcal{F}^{\textnormal{new}}_{\textnormal{LP}}}}
\newcommand{\FeasibleIP}{\ensuremath{\mathcal{F}_{\textnormal{IP}}}}


\makeatletter
\newcommand{\removelatexerror}{\let\@latex@error\@gobble}
\makeatother

\newcounter{ipCounter}
\NewDocumentEnvironment{IPFormulation}{m}{%
\refstepcounter{ipCounter}
\begin{algorithm}[#1]%
\renewcommand\thealgocf{\arabic{ipCounter}}
}{%
\end{algorithm}
\addtocounter{algocf}{-1}
}

\NewDocumentEnvironment{IPFormulationStar}{m}{%
\refstepcounter{ipCounter}
\begin{algorithm*}[#1]%
\renewcommand\thealgocf{\arabic{ipCounter}}
}{%
\end{algorithm*}
\addtocounter{algocf}{-1}
}



\newcommand{\spaceSolReq}[1][\req]{\ensuremath{\mathcal{M}_{#1}}}
\newcommand{\spaceLP}{\ensuremath{\mathcal{F}^{\textnormal{mcf}}_{\textnormal{LP}}}}
\newcommand{\spaceIP}{\ensuremath{\mathcal{F}^{\textnormal{mcf}}_{\textnormal{IP}}}}
\newcommand{\spaceIPCC}{\ensuremath{\mathcal{F}^{\textnormal{}}_{\textnormal{IP}}}}
\newcommand{\spaceLPCC}{\ensuremath{\mathcal{F}^{\textnormal{}}_{\textnormal{LP}}}}
\newcommand{\spaceIPMDK}{\ensuremath{\mathcal{F}^{\textnormal{MDK}}_{\textnormal{IP}}}}
\newcommand{\spaceLPMDK}{\ensuremath{\mathcal{F}^{\textnormal{MDK}}_{\textnormal{LP}}}}
\newcommand{\spaceLPNew}{\ensuremath{\mathcal{F}^{\textnormal{new}}_{\textnormal{LP}}}}
\newcommand{\spaceLPD}{\ensuremath{\mathcal{F}^{\mathcal{D}}_{\textnormal{LP}}}}
\newcommand{\spaceLPDreq}[1][\req]{\ensuremath{\mathcal{F}^{\mathcal{D}}_{\textnormal{LP},#1}}}

\DeclareDocumentCommand{\OptProfit}{}{\ensuremath{\hat{P}}}


\DeclareDocumentCommand{\NodeG}{O{\req} O{\pi}}{\ensuremath{G^N_{#1,#2}}}
\DeclareDocumentCommand{\NodeV}{O{\req} O{\pi}}{\ensuremath{V^N_{#1,#2}}}
\DeclareDocumentCommand{\NodeE}{O{\req} O{\pi}}{\ensuremath{E^N_{#1,#2}}}

\DeclareDocumentCommand{\EdgeG}{O{\req} O{i} O{j} O{u}}{\ensuremath{G^E_{#1,#2,#3,#4}}}
\DeclareDocumentCommand{\EdgeV}{O{\req} O{i} O{j} O{u}}{\ensuremath{V^E_{#1,#2,#3,#4}}}
\DeclareDocumentCommand{\EdgeE}{O{\req} O{i} O{j} O{u}}{\ensuremath{E^E_{#1,#2,#3,#4}}}

\DeclareDocumentCommand{\VESD}{O{\req}}{\ensuremath{\overrightarrow{E}_{#1}}}
\DeclareDocumentCommand{\VEOD}{O{\req}}{\ensuremath{\overleftarrow{E}_{#1}}}
\DeclareDocumentCommand{\VESigmaD}{O{\req} O{\sigma}}{\ensuremath{E_{#1,#2}}}

\DeclareDocumentCommand{\NodeVRange}{O{\req} O{i} O{j}}{\ensuremath{V^N_{#1,#2,#3}}}
\DeclareDocumentCommand{\NodeVRangeRange}{O{\req} O{i} O{j} O{u} O{v}}{\ensuremath{V^N_{#1,#2,#3,#4,#5}}}

\DeclareDocumentCommand{\path}{O{k }}{\ensuremath{P_{#1}}}
\DeclareDocumentCommand{\cycle}{O{k }}{\ensuremath{C_{#1}}}

\DeclareDocumentCommand{\NodePaths}{O{\req}}{\ensuremath{\mathcal{P}_{#1}}}


\DeclareDocumentCommand{\loadV}{O{\req} O{u}}{\ensuremath{l_{#1,#2}}}
\DeclareDocumentCommand{\loadE}{O{\req} O{u} O{v}}{\ensuremath{l_{#1,#2,#3}}}
\DeclareDocumentCommand{\loadX}{O{\req} O{x}}{\ensuremath{l_{#1,#2}}}
\DeclareDocumentCommand{\decomp}{O{\req} O{k}}{\ensuremath{D_{#1}^{#2}}}
\DeclareDocumentCommand{\decompHat}{O{\req} O{k}}{\ensuremath{{\hat{D}}_{#1}^{#2}}}
\DeclareDocumentCommand{\load}{O{\req} O{k}}{\ensuremath{l_{#1}^{#2}}}
\DeclareDocumentCommand{\prob}{O{\req} O{k}}{\ensuremath{f_{#1}^{#2}}}
\DeclareDocumentCommand{\mapping}{O{\req} O{k}}{\ensuremath{m_{#1}^{#2}}}

\DeclareDocumentCommand{\loadHat}{O{\req} O{k}}{\ensuremath{\hat{l}_{#1}^{#2}}}
\DeclareDocumentCommand{\probHat}{O{\req} O{k}}{\ensuremath{\hat{f}_{#1}^{#2}}}
\DeclareDocumentCommand{\mappingHat}{O{\req} O{k}}{\ensuremath{\hat{m}_{#1}^{#2}}}

\DeclareDocumentCommand{\loadHat}{O{\req} O{k}}{\ensuremath{\hat{l}_{#1}^{#2}}}
\DeclareDocumentCommand{\probHat}{O{\req} O{k}}{\ensuremath{\hat{f}_{#1}^{#2}}}
\DeclareDocumentCommand{\mappingHat}{O{\req} O{k}}{\ensuremath{\hat{m}_{#1}^{#2}}}

\DeclareDocumentCommand{\Exp}{}{\ensuremath{\mathbb{E}}}
\DeclareDocumentCommand{\randVarX}{O{\req} O{k}}{\ensuremath{X_{#1}^{#2}}}
\DeclareDocumentCommand{\randVarY}{O{\req}}{\ensuremath{Y_{#1}}}
\DeclareDocumentCommand{\randVarZ}{O{\req}}{\ensuremath{Z_{#1}}}
\DeclareDocumentCommand{\randVarL}{O{x}}{\ensuremath{L_{x}}}
\DeclareDocumentCommand{\randVarLX}{O{\req} O{x} O{y}}{\ensuremath{L_{#1,#2,#3}}}
\DeclareDocumentCommand{\randVarLNode}{O{\req} O{\type} O{u}}{\ensuremath{L_{#1,#2,#3}}}
\DeclareDocumentCommand{\randVarLEdge}{O{\req} O{u} O{v}}{\ensuremath{L_{#1,#2,#3}}}
\DeclareDocumentCommand{\randVarM}{O{\req}}{\ensuremath{M_{#1}}}
\DeclareDocumentCommand{\randVarC}{O{\req}}{\ensuremath{C_{#1}}}

\DeclareDocumentCommand{\ProbVarX}{O{1}}{\ensuremath{\mathbb{P}(\randVarX = #1)}}
\DeclareDocumentCommand{\ProbVarY}{O{1}}{\ensuremath{\mathbb{P}(\randVarY = #1)}}
\DeclareDocumentCommand{\ProbVarZ}{O{1}}{\ensuremath{\mathbb{P}(\randVarZ = #1)}}
\DeclareDocumentCommand{\ProbVarL}{O{1}}{\ensuremath{\mathbb{P}(\randVarL = #1)}}
\DeclareDocumentCommand{\ProbVarM}{O{1}}{\ensuremath{\mathbb{P}(\randVarM = #1)}}
\DeclareDocumentCommand{\ProbVarC}{O{1}}{\ensuremath{\mathbb{P}(\randVarC = #1)}}

\DeclareDocumentCommand{\randVarObjApprox}{}{\ensuremath{Obj_{\textnormal{Alg}}}}
\DeclareDocumentCommand{\optLP}{}{\ensuremath{\textnormal{B}_{\textnormal{LP}}}}
\DeclareDocumentCommand{\optLPstd}{}{\ensuremath{\textnormal{Opt}^{\textnormal{mcf}}_{\textnormal{LP}}}}
\DeclareDocumentCommand{\optLPnew}{}{\ensuremath{\textnormal{Opt}^{\textnormal{new}}_{\textnormal{LP}}}}
\DeclareDocumentCommand{\optIP}{}{\ensuremath{\textnormal{B}_{\textnormal{IP}}}}

\DeclareDocumentCommand{\WAC}{O{\req}}{\ensuremath{\textnormal{WC}_{\req}}}


\DeclareDocumentCommand{\PotEmbeddings}{O{\req}}{\ensuremath{\mathcal{D}_{#1}}}
\DeclareDocumentCommand{\PotEmbeddingsHat}{O{\req}}{\ensuremath{\hat{\mathcal{D}}_{#1}}}


\DeclareDocumentCommand{\maxLoadX}{O{x}}{\ensuremath{\textnormal{max}^{L,\sum}_{#1}}}
\DeclareDocumentCommand{\maxLoadV}{O{\req} O{\type} O{u}}{\ensuremath{\textnormal{max}^{L}_{#1,#2,#3}}}
\DeclareDocumentCommand{\maxLoadE}{O{\req} O{u} O{v}}{\ensuremath{\textnormal{max}^{L}_{#1,#2,#3}}}
\DeclareDocumentCommand{\maxLoadVSum}{O{\req} O{\type} O{u}}{\ensuremath{\textnormal{max}^{L,\Sigma}_{#1,#2,#3}}}
\DeclareDocumentCommand{\maxLoadESum}{O{u} O{v}}{\ensuremath{\textnormal{max}^{L,\Sigma}_{#1,#2}}}

\DeclareDocumentCommand{\DeltaV}{}{\ensuremath{\Delta_V}}
\DeclareDocumentCommand{\DeltaE}{}{\ensuremath{\Delta_E}}


\DeclareDocumentCommand{\VVroot}{O{\req}}{\ensuremath{r_{#1}}}
\DeclareDocumentCommand{\VVpred}{O{\req}}{\ensuremath{\pi_{#1}}}

\newcommand{\VGbfs}[1][\req]{\ensuremath{G^{\mathcal{A}}_{#1}}}
\newcommand{\VVbfs}[1][\req]{\ensuremath{V^{\mathcal{A}}_{#1}}}
\newcommand{\VEbfs}[1][\req]{\ensuremath{E^{\mathcal{A}}_{#1}}}

\DeclareDocumentCommand{\Cycles}{O{\req}}{\ensuremath{\mathcal{C}_{#1}}}
\DeclareDocumentCommand{\Paths}{O{\req}}{\ensuremath{\mathcal{P}_{#1}}}

\DeclareDocumentCommand{\VVcycleSource}{O{\req} O{k}}{\ensuremath{s^{C_{#2}}_{#1}}}
\DeclareDocumentCommand{\VVcycleTarget}{O{\req} O{k}}{\ensuremath{t^{C_{#2}}_{#1}}}
\DeclareDocumentCommand{\VVpathSource}{O{\req} O{k}}{\ensuremath{s^{P_{#2}}_{#1}}}
\DeclareDocumentCommand{\VVpathTarget}{O{\req} O{k}}{\ensuremath{t^{P_{#2}}_{#1}}}

\DeclareDocumentCommand{\VGcycle}{O{\req} O{k}}{\ensuremath{G^{\mathcal{A},C_{#2}}_{#1}}}
\DeclareDocumentCommand{\VVcycle}{O{\req} O{k}}{\ensuremath{V^{\mathcal{A},C_{#2}}_{#1}}}
\DeclareDocumentCommand{\VEcycle}{O{\req} O{k}}{\ensuremath{E^{\mathcal{A},C_{#2}}_{#1}}}
\DeclareDocumentCommand{\VEcycleSame}{O{\req} O{k}}{\ensuremath{\overrightarrow{E}^{C_{#2}}_{#1}}}
\DeclareDocumentCommand{\VEcycleDiff}{O{\req} O{k}}{\ensuremath{\overleftarrow{E}^{C_{#2}}_{#1}}}

\DeclareDocumentCommand{\VGcycleOrig}{O{\req} O{k}}{\ensuremath{G^{C_{#2}}_{#1}}}
\DeclareDocumentCommand{\VVcycleOrig}{O{\req} O{k}}{\ensuremath{V^{C_{#2}}_{#1}}}
\DeclareDocumentCommand{\VEcycleOrig}{O{\req} O{k}}{\ensuremath{E^{C_{#2}}_{#1}}}
\DeclareDocumentCommand{\VGpath}{O{\req} O{k}}{\ensuremath{G^{P_{#2}}_{#1}}}
\DeclareDocumentCommand{\VVpath}{O{\req} O{k}}{\ensuremath{V^{P_{#2}}_{#1}}}
\DeclareDocumentCommand{\VEpath}{O{\req} O{k}}{\ensuremath{E^{P_{#2}}_{#1}}}

\DeclareDocumentCommand{\VEpathSame}{O{\req} O{k}}{\ensuremath{\overrightarrow{E}^{P_{#2}}_{#1}}}
\DeclareDocumentCommand{\VEpathDiff}{O{\req} O{k}}{\ensuremath{\overleftarrow{E}^{P_{#2}}_{#1}}}

\DeclareDocumentCommand{\VEDiff}{O{\req}}{\ensuremath{\overleftarrow{E}^{\mathcal{A}}_{#1}}}

\DeclareDocumentCommand{\VEcycles}{O{\req}}{\ensuremath{E^{\mathcal{C}}_{#1}}}
\DeclareDocumentCommand{\VEpaths}{O{\req}}{\ensuremath{E^{\mathcal{P}}_{#1}}}

\DeclareDocumentCommand{\VVcycleSourcesTargets}{O{\req}}{\ensuremath{V^{\mathcal{C},\pm}_{#1}}}
\DeclareDocumentCommand{\VVpathSourcesTargets}{O{\req}}{\ensuremath{V^{\mathcal{P},\pm}_{#1}}}
\DeclareDocumentCommand{\VVSourcesTargets}{O{\req}}{\ensuremath{V^{\pm}_{#1}}}

\DeclareDocumentCommand{\VVcycleSources}{O{\req}}{\ensuremath{V^{\mathcal{C},+}_{#1}}}
\DeclareDocumentCommand{\VVpathSources}{O{\req}}{\ensuremath{V^{\mathcal{P},+}_{#1}}}

\DeclareDocumentCommand{\VVcycleTargets}{O{\req}}{\ensuremath{V^{\mathcal{C},-}_{#1}}}
\DeclareDocumentCommand{\VVpathTargets}{O{\req}}{\ensuremath{V^{\mathcal{P},-}_{#1}}}

\DeclareDocumentCommand{\VGcycleBranchR}{O{\req} O{k}}{\ensuremath{G^{C_{#2}, B_1}_{#1}}}
\DeclareDocumentCommand{\VVcycleBranchR}{O{\req} O{k}}{\ensuremath{V^{C_{#2}, B_1}_{#1}}}
\DeclareDocumentCommand{\VEcycleBranchR}{O{\req} O{k}}{\ensuremath{E^{C_{#2}, B_1}_{#1}}}

\DeclareDocumentCommand{\VGcycleBranchL}{O{\req} O{k}}{\ensuremath{G^{C_{#2}, B_2}_{#1}}}
\DeclareDocumentCommand{\VVcycleBranchL}{O{\req} O{k}}{\ensuremath{V^{C_{#2}, B_2}_{#1}}}
\DeclareDocumentCommand{\VEcycleBranchL}{O{\req} O{k}}{\ensuremath{E^{C_{#2}, B_2}_{#1}}}

\DeclareDocumentCommand{\VGdecomp}{O{\req} O{k}}{\ensuremath{G^{\mathcal{D}}_{#1}}}
\DeclareDocumentCommand{\VVdecomp}{O{\req} O{k}}{\ensuremath{V^{\mathcal{D}}_{#1}}}
\DeclareDocumentCommand{\VEdecomp}{O{\req} O{k}}{\ensuremath{E^{\mathcal{D}}_{#1}}}

\DeclareDocumentCommand{\VVbranching}{O{\req} }{\ensuremath{\mathcal{B}_{#1}}}
\DeclareDocumentCommand{\VVbranchingcycle}{O{\req} O{k}}{\ensuremath{\mathcal{B}^{C_{#2}}_{#1}}}
\DeclareDocumentCommand{\VVbranchingpath}{O{\req} O{k}}{\ensuremath{\mathcal{B}^{P_{k}}_{#1}}}
\DeclareDocumentCommand{\VVjoin}{O{\req} }{\ensuremath{\mathcal{J}_{#1}}}
\DeclareDocumentCommand{\VVaggregation}{O{\req} }{\ensuremath{\mathcal{A}_{#1}}}

\DeclareDocumentCommand{\VGextcycle}{O{\req} O{k}}{\ensuremath{G^{C_{#2}}_{#1,\textnormal{ext}}}}

\DeclareDocumentCommand{\VVextcycle}{O{\req} O{k}}{\ensuremath{V^{C_{#2}}_{#1,\textnormal{ext}}}}
\DeclareDocumentCommand{\VVextcycleSources}{O{\req} O{k}}{\ensuremath{V^{C_{#2}}_{#1,+}}}
\DeclareDocumentCommand{\VVextcycleTargets}{O{\req} O{k}}{\ensuremath{V^{C_{#2}}_{#1,-}}}
\DeclareDocumentCommand{\VVextcycleSubstrate}{O{\req} O{k}}{\ensuremath{V^{C_{#2}}_{#1,S}}}

\DeclareDocumentCommand{\VEextcycle}{O{\req} O{k}}{\ensuremath{E^{C_{#2}}_{#1,\textnormal{ext}}}}
\DeclareDocumentCommand{\VEextcycleSources}{O{\req} O{k}}{\ensuremath{E^{C_{#2}}_{#1,+}}}
\DeclareDocumentCommand{\VEextcycleTargets}{O{\req} O{k}}{\ensuremath{E^{C_{#2}}_{#1,-}}}
\DeclareDocumentCommand{\VEextcycleSubstrate}{O{\req} O{k}}{\ensuremath{E^{C_{#2}}_{#1,S}}}
\DeclareDocumentCommand{\VEextcycleF}{O{\req} O{k}}{\ensuremath{E^{C_{#2}}_{#1,F}}}

\DeclareDocumentCommand{\VGextpath}{O{\req} O{k}}{\ensuremath{G^{{P_{#2}}}_{#1,\textnormal{ext}}}}

\DeclareDocumentCommand{\VVextpath}{O{\req} O{k}}{\ensuremath{V^{P_{#2}}_{#1,\textnormal{ext}}}}
\DeclareDocumentCommand{\VVextpathSources}{O{\req} O{k}}{\ensuremath{V^{P_{#2}}_{#1,+}}}
\DeclareDocumentCommand{\VVextpathTargets}{O{\req} O{k}}{\ensuremath{V^{P_{#2}}_{#1,-}}}
\DeclareDocumentCommand{\VVextpathSubstrate}{O{\req} O{k}}{\ensuremath{V^{P_{#2}}_{#1,S}}}

\DeclareDocumentCommand{\VEextpath}{O{\req} O{k}}{\ensuremath{E^{P_{#2}}_{#1,\textnormal{ext}}}}
\DeclareDocumentCommand{\VEextpathSources}{O{\req} O{k}}{\ensuremath{E^{P_{#2}}_{#1,+}}}
\DeclareDocumentCommand{\VEextpathTargets}{O{\req} O{k}}{\ensuremath{E^{P_{#2}}_{#1,-}}}
\DeclareDocumentCommand{\VEextpathSubstrate}{O{\req} O{k}}{\ensuremath{E^{P_{#2}}_{#1,S}}}
\DeclareDocumentCommand{\VEextpathF}{O{\req} O{k}}{\ensuremath{E^{P_{#2}}_{#1,F}}}

\DeclareDocumentCommand{\forest}{O{\req}}{\ensuremath{\mathcal{F}_{#1}}}

\DeclareDocumentCommand{\VGforest}{O{\req}}{\ensuremath{G^{\mathcal{A},\mathcal{F}}_{#1}}}
\DeclareDocumentCommand{\VVforest}{O{\req}}{\ensuremath{V^{\mathcal{A},\mathcal{F}}_{#1}}}
\DeclareDocumentCommand{\VEforest}{O{\req}}{\ensuremath{E^{\mathcal{A},\mathcal{F}}_{#1}}}

\DeclareDocumentCommand{\VGforestOrig}{O{\req}}{\ensuremath{G^{\mathcal{F}}_{#1}}}
\DeclareDocumentCommand{\VVforestOrig}{O{\req}}{\ensuremath{V^{\mathcal{F}}_{#1}}}
\DeclareDocumentCommand{\VEforestOrig}{O{\req}}{\ensuremath{E^{\mathcal{F}}_{#1}}}


\DeclareDocumentCommand{\varFlowInput}{O{\req} O{i} O{u}}{\ensuremath{f^+_{#1,#2,#3}}}
\DeclareDocumentCommand{\varFlowOutput}{O{\req} O{i} O{u}}{\ensuremath{f^+_{#1,#2,#3}}}

\DeclareDocumentCommand{\VEextcycleHorizontal}{O{\req} O{k} O{u} O{v}}{\ensuremath{E^{C_{#2}}_{#1,\textnormal{ext},#3,#4}}}
\DeclareDocumentCommand{\VEextpathHorizontal}{O{\req} O{k} O{u} O{v}}{\ensuremath{E^{P_{#2}}_{#1,\textnormal{ext},#3,#4}}}
\DeclareDocumentCommand{\VEextcycleVertical}{O{\req} O{k} O{\type} O{u}}{\ensuremath{E^{C_{#2}}_{#1,\textnormal{ext},#3,#4}}}
\DeclareDocumentCommand{\VEextpathVertical}{O{\req} O{k} O{\type} O{u}}{\ensuremath{E^{P_{#2}}_{#1,\textnormal{ext},#3,#4}}}

\DeclareDocumentCommand{\VEextCGHorizontal}{O{\req} O{u} O{v}}{\ensuremath{E^{\textnormal{ext,SCG}}_{#1,#2,#3}}}
\DeclareDocumentCommand{\VEextCGVertical}{O{\req} O{\type} O{u}}{\ensuremath{E^{\textnormal{ext,SCG}}_{#1,#2,#3}}}

\DeclareDocumentCommand{\VVextCGFlowNodes}{O{\req}}{\ensuremath{V^{\textnormal{ext,SCG}}_{#1,\textnormal{flow}}}}
\DeclareDocumentCommand{\VEextCGFlowEdges}{O{\req}}{\ensuremath{E^{\textnormal{ext,SCG}}_{#1,\textnormal{flow}}}}


\DeclareDocumentCommand{\Queue}{}{\ensuremath{\mathcal{Q}}}
\DeclareDocumentCommand{\QueueC}{}{\ensuremath{\mathcal{Q}_{\mathcal{C}}}}
\DeclareDocumentCommand{\QueueP}{}{\ensuremath{\mathcal{Q}_{\mathcal{P}}}}
\DeclareDocumentCommand{\UsedPaths}{}{\ensuremath{\mathcal{P}}}
\DeclareDocumentCommand{\Variables}{}{\ensuremath{\mathcal{V}}}

\DeclareDocumentCommand{\VGextcycleFlow}{O{\req} O{k}}{\ensuremath{G^{C_{#2}}_{#1,\textnormal{ext},f}}}

\DeclareDocumentCommand{\VGextcycleFlowBranchR}{O{\req} O{k}}{\ensuremath{G^{C_{#2},{B}_1}_{#1,\textnormal{ext},f}}}
\DeclareDocumentCommand{\VVextcycleFlowBranchR}{O{\req} O{k}}{\ensuremath{V^{C_{#2},{B}_1}_{#1,\textnormal{ext},f}}}
\DeclareDocumentCommand{\VEextcycleFlowBranchR}{O{\req} O{k}}{\ensuremath{E^{C_{#2},{B}_1}_{#1,\textnormal{ext},f}}}

\DeclareDocumentCommand{\VGextcycleFlowBranchL}{O{\req} O{k}}{\ensuremath{G^{C_{#2},{B}_2}_{#1,\textnormal{ext},f}}}
\DeclareDocumentCommand{\VVextcycleFlowBranchL}{O{\req} O{k}}{\ensuremath{V^{C_{#2},{B}_2}_{#1,\textnormal{ext},f}}}
\DeclareDocumentCommand{\VEextcycleFlowBranchL}{O{\req} O{k}}{\ensuremath{E^{C_{#2},{B}_2}_{#1,\textnormal{ext},f}}}

\DeclareDocumentCommand{\VGextcycleBranchR}{O{\req} O{k}}{\ensuremath{G^{C_{#2},{B}_1}_{#1,\textnormal{ext}}}}
\DeclareDocumentCommand{\VVextcycleBranchR}{O{\req} O{k}}{\ensuremath{V^{C_{#2},{B}_1}_{#1,\textnormal{ext}}}}
\DeclareDocumentCommand{\VEextcycleBranchR}{O{\req} O{k}}{\ensuremath{E^{C_{#2},{B}_1}_{#1,\textnormal{ext}}}}

\DeclareDocumentCommand{\VGextcycleBranchL}{O{\req} O{k}}{\ensuremath{G^{C_{#2},{B}_2}_{#1,\textnormal{ext}}}}
\DeclareDocumentCommand{\VVextcycleBranchL}{O{\req} O{k}}{\ensuremath{V^{C_{#2},{B}_2}_{#1,\textnormal{ext}}}}
\DeclareDocumentCommand{\VEextcycleBranchL}{O{\req} O{k}}{\ensuremath{E^{C_{#2},{B}_2}_{#1,\textnormal{ext}}}}

\DeclareDocumentCommand{\VGextpathFlow}{O{\req} O{k}}{\ensuremath{G^{P_{#2}}_{#1,\textnormal{ext},f}}}
\DeclareDocumentCommand{\VVextpathFlow}{O{\req} O{k}}{\ensuremath{V^{P_{#2}}_{#1,\textnormal{ext},f}}}
\DeclareDocumentCommand{\VEextpathFlow}{O{\req} O{k}}{\ensuremath{E^{P_{#2}}_{#1,\textnormal{ext},f}}}

\DeclareDocumentCommand{\VVKSource}{O{\req} O{K}}{\ensuremath{s^{K}_{#1}}}
\DeclareDocumentCommand{\VVKTarget}{O{\req} O{K}}{\ensuremath{t^{K}_{#1}}}
\DeclareDocumentCommand{\VVKSourcesTargets}{O{\req}}{\ensuremath{V^{K,\pm}_{#1}}}

\newcommand{\fracSol}{\ensuremath{\mathsf{opt}_{\textnormal{LP}}}}
\newcommand{\intSol}{\ensuremath{\mathsf{opt}_{\textnormal{IP}}}}
\newcommand{\Prob}{\ensuremath{\mathsf{Pr}}}

\newcommand{\MIPMCF}{\ensuremath{\mathrm{MIP}_{\mathrm{MCF}}}}
\newcommand{\LPNOVEL}{\ensuremath{\mathrm{LP}_{\mathrm{novel}}}}
\newcommand{\RRMIN}{\ensuremath{\mathrm{RR}_{\mathrm{MinLoad}}}}
\newcommand{\RRMAX}{\ensuremath{\mathrm{RR}_{\mathrm{MaxProfit}}}}
\newcommand{\RRHEUR}{\ensuremath{\mathrm{RR}_{\mathrm{Heuristic}}}}

\newcommand{\customParagraphStar}[1]{\subsection{{#1}}}

\begin{abstract}

The Virtual Network Embedding Problem (VNEP) captures the essence of many resource allocation problems of today's infrastructure providers, which offer their physical computation and networking resources to customers. Customers request resources in the form of Virtual Networks, i.e.~as a directed graph which specifies computational requirements at the nodes and communication requirements on the edges.
An embedding of a Virtual Network on the shared physical infrastructure is the joint mapping of (virtual) nodes to physical servers together with the mapping of (virtual) edges onto paths in the physical network connecting the respective servers. 

This work initiates the study of approximation algorithms for the VNEP. Concretely, we study the offline setting with admission control: given multiple request graphs the task is to embed the most profitable subset while not exceeding resource capacities.
Our approximation is based on the randomized rounding of Linear Programming (LP) solutions. Interestingly, we uncover that the standard LP formulation for the VNEP exhibits an inherent structural deficit when considering general virtual network topologies: its solutions cannot be decomposed into valid embeddings. In turn, focusing on the class of cactus request graphs, we devise a novel LP formulation, whose solutions can be decomposed into convex combinations of valid embedding. Proving performance guarantees of our rounding scheme, we obtain the first approximation algorithm for the VNEP in the resource augmentation model. 

We propose two types of rounding heuristics and evaluate their performance in an extensive computational study. Our results indicate that randomized rounding can yield good solutions (even without augmentations). Specifically, heuristic rounding achieves 73.8\% of the baseline's profit, while not exceeding capacities.

\end{abstract}

\maketitle

\section{Introduction}
Cloud applications usually consist of multiple distributed components
(e.g., virtual machines, containers), 
which results in substantial communication requirements. If the provider fails to ensure that these communication requirements are met, the performance can suffer dramatically~\cite{talkabout}. Consequently, over the last years,
several proposals have been made to \emph{jointly} provision the computational functionality \emph{together} with appropriate network resources. The Virtual Network Embedding Problem (VNEP) captures the core of this problem: given a directed graph specifying computational requirements at the nodes and bandwidth requirements on the edges, an embedding of this \emph{Virtual Network} in the physical network has to be found, such that both the computational and the network requirements are met. 
Figure~\ref{fig:vnet-examples-sc-vc} illustrates two incarnations of virtual networks: \emph{service chains}~\cite{mehraghdam2014specifying} and \emph{virtual clusters}~\cite{oktopus}. 

\begin{figure}[tbhp]
\centering
\includegraphics[width=0.8\columnwidth]{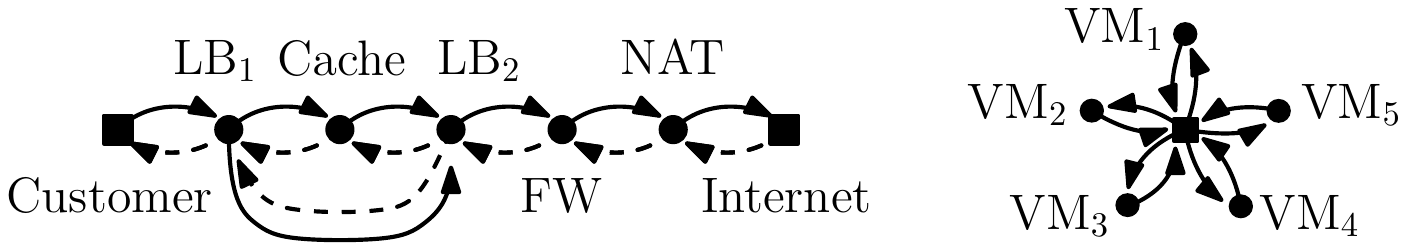}
\caption{Examples for virtual networks `in the wild'. The left graph shows a service chain for mobile operators~\cite{ietf-sfc-use-case-mobility-06}: load-balancers route (parts of the) traffic through a cache. Furthermore,  a firewall and a network-address translation are used.
 The right graph depicts the Virtual Cluster abstraction for provisioning virtual machines (VMs) in data centers. The abstraction provides connectivity guarantees via a \emph{logical switch} in the center~\cite{oktopus}.}
\label{fig:vnet-examples-sc-vc}
\vspace{-6pt}
\end{figure}

We study the offline setting with admission control: given multiple requests the task is to embed the most profitable subset while not exceeding resource capacities.

\subsection{Formal Problem Statement}

\DeclareDocumentCommand{\maxDemandV}{O{\req} O{\type} O{u}}{\ensuremath{{d}_{\textnormal{max}}(#1,#2,#3)}}
\DeclareDocumentCommand{\maxDemandE}{O{\req} O{u} O{v}}{\ensuremath{d_{\textnormal{max}} ({#1,#2,#3})}}
\DeclareDocumentCommand{\maxDemandX}{O{\req} O{x} O{y}}{\ensuremath{d_{\textnormal{max}} ({#1,#2,#3})}}
\DeclareDocumentCommand{\maxAllocV}{O{\req} O{\type} O{u}}{\ensuremath{{A_{\textnormal{max}}}({#1,#2,#3})}}
\DeclareDocumentCommand{\maxAllocE}{O{\req} O{u} O{v}}{\ensuremath{ {A_{\textnormal{max}}}({#1,#2,#3})}}
\DeclareDocumentCommand{\maxAllocX}{O{\req} O{x} O{y}}{\ensuremath{ {A_{\textnormal{max}}}({#1,#2,#3})}}

In the light of the recent interest in Service Chaining~\cite{mehraghdam2014specifying}, we extend the VNEP's general definition~\cite{vnep-survey} by considering different \emph{types} of computational nodes. We refer to the physical network as the \emph{substrate network}.
The substrate $\SG=(\SV,\SE)$ is offering a set~$\types$ of computational types. This set of types may contain, e.g., `FW'~(firewall), `x86 server', etc. For a type $\type \in \types$, the set $\SVTypes \subseteq \SV$ denotes the substrate nodes that can host functionality of type $\type$. Denoting the node resources by $\SRV = \{(\type, u)~| \type \in \types, u \in \SVTypes\}$ and all substrate resources by $\SR = \SRV \cup \SE$, the capacity of nodes and edges is denoted by $\Scap(x,y) > 0$ for  $(x,y) \in \SR$.

For each request $\req \in \requests$, a directed graph $\VG=(\VV,\VE)$ together with a profit $\Vprofit > 0$ is given. We refer to the respective nodes as virtual or request nodes and similarly refer to the respective edges as virtual or request edges.
The types of virtual nodes are indicated by the function $\Vtype : \VV \to \types$. 

Based on policies of the customer or the provider, the mapping of virtual node $i \in \VV$ is restricted to a set $\VVloc \subseteq \SVTypes[\Vtype(i)]$, while the mapping of virtual edge $(i,j)$ is restricted to a subset of substrate edges $\VEloc \subseteq \SE$. 
Each virtual node $i \in \VV$ and each edge $(i,j) \in \VE$ is attributed with a resource demand $\Vcap(i) \geq 0$ and $\Vcap(i,j) \geq 0$, respectively. 
Virtual nodes and edges can only be mapped on substrate nodes and edges of sufficient capacity, i.e. \mbox{$\VVloc \subseteq \{u \in \SVTypes[\Vtype(i)] | \Scap(u) \geq \Vcap(i)\}$} and \mbox{$\VEloc \subseteq \{(u,v) \in \SE | \Scap(u,v) \geq \Vcap(i,j)\}$} holds.

We denote by $d_{\max}(r,x,y)$ the maximal demand that a request $\req$ may impose on a resource $(x,y) \in \SR$:
{
\small\begin{alignat*}{5}
\maxDemandV &= && \max ( \{0\} \cup \{\Vcap(i) | i \in \VV: \type(i) = \tau \wedge u \in \VVloc\} ) \\
\maxDemandE &= && \max ( \{0\} \cup \{ \Vcap(i,j) | (i,j) \in  \VE: (u,v) \in \VEloc \})
\end{alignat*}
}
In the following the notions of valid mappings (respecting mapping constraints) and feasible embeddings (respecting resource constraints) are introduced to formalize the VNEP.

\begin{definition}[Valid Mapping]
\label{def:valid-mapping}
A valid mapping~$\map$ of request~$\req \in \requests$ is a tuple~$(\mapV, \mapE)$ of functions $\mapV : \VV \to \SV$ and $\mapE : \VE \to \mathcal{P}(\SE)$, such that the following holds:
\begin{itemize}
\item Virtual nodes are mapped to allowed substrate nodes: $\mapV(i) \in  \VVloc$  holds for all $i \in \VV$.
\item The mapping $\mapE(i,j)$ of virtual edge $(i,j) \in \VE$ is an edge-path connecting~$\mapV(i)$ to $\mapV(j)$ only using allowed edges, i.e. $\mapE(i,j) \subseteq \mathcal{P}(\VEloc)$ holds. 
\end{itemize} 
We denote by $\spaceSolReq$ the set of valid mappings of request $\req \in \requests$.
\end{definition}

\begin{definition}[Allocations of Valid Mappings] We denote by $A(\map,x,y)$ the cumulative allocation induced by the valid mapping $\map \in \spaceSolReq$ 
on resource $(x,y) \in \SR$:
\begin{alignat*}{4}
A(\map, \tau,u) &= && \sum \nolimits_{i \in \VV, \type(i)=\tau, \mapV(i)=u } \Vcap(i) & ~& \forall (\tau,u) \in \SRV \\
A(\map,u,v) &= && \sum \nolimits_{ (i,j)\in \VE, (u,v) \in \mapE(i,j)}   \Vcap(i,j) &~& \forall (u,v) \in \SE
\end{alignat*}
The maximal allocation that a valid mapping of request $\req \in \requests$ may impose on a substrate resource $(x,y) \in \SR$ is denoted by $\maxAllocX  =  \max_{\map \in \spaceSolReq} A(\map,x,y)$.
\end{definition}

\begin{definition}[Feasible Embedding]
\label{def:feasible-embedding}
A feasible embedding of a subset of requests $\requestsP \subseteq \requests$ is a collection of valid mappings $\{\map[\req]\}_{r \in \requestsP}$, such that the \emph{cumulative} allocations on nodes and edges does not exceed the substrate capacities, i.e. $\sum_{\req \in \requestsP} A(\map,x,y) \leq \Scap(x,y)$ holds for $(x,y) \in \SR$.
\end{definition}

\begin{definition}[Virtual Network Embedding Problem]
\label{def:scep-with-admission-control}
The VNEP asks for a feasible embedding $\{m_{\req}\}_{\req \in \requestsP}$ of a subset of requests $\requestsP \subseteq \requests$ maximizing the profit $\sum_{\req \in \requestsP} \Vprofit[\req]$.
\end{definition}

\subsection{Related Work}

In the last decade, the VNEP has attracted much attention due to its many applications and the survey~\cite{vnep-survey} from 2013 already lists more than 80 different algorithms for its many variations~\cite{vnep-survey}. 
The VNEP is known to be $\mathcal{NP}$-hard and inapproximable in general (unless $\mathcal{P} = \mathcal{NP}$)~\cite{rostSchmidVNEPComplexity}.
Based on the hardness of the VNEP, most works consider heuristics without any performance guarantee~\cite{vnep-survey,vnep}. 
Other works proposed exact methods as integer or constraint programming, coming at the cost of an exponential runtime~\cite{HVSBFTF15,jarray2015decomposition,rostSchmidFeldmann2014}. 

A column generation approach was proposed by Jarray et al. in~\cite{jarray2015decomposition} to efficiently compute solutions to the VNEP by generating feasible mappings `on-the-fly subject to a specific cost measure. In particular, Jarray et al. compute feasible mappings by relying on heuristics and, if no feasible heuristical solution was found, rely on Mixed-Integer Programming to compute a cost optimal feasible embedding in non-polynomial time.
We believe that our work can bridge the gap between the heuristic generation of feasible mappings and the optimal generation of feasible mappings as our approach can be used to compute \emph{approximate} mappings.

Acknowledging the hardness of the general VNEP and the diversity of applications, several subproblems of the VNEP have been studied recently by considering restricted graph classes for the virtual networks and the substrate graph. For example, virtual clusters with uniform demands  are studied in \cite{oktopus,ccr15emb}, line requests are studied in~\cite{sss-guy,sirocco16path,sirocco15} and tree requests were studied in \cite{sirocco16path,bansal2011minimum}.

Considering approximation algorithms, Even et al. employed randomized rounding in~\cite{sirocco16path} to obtain a constant approximation for embedding line requests on arbitrary substrate graphs under strong assumptions on \emph{both} the benefits and the capacities.
In their interesting work, 
Bansal et al.~\cite{bansal2011minimum} 
give an
$n^{O(d)}$ time~$O(d^2 \log{(nd)})$-approximation algorithm 
for minimizing the load of embedding $d$-depth trees based on a $n$-node substrate. Their result is based on a strong LP relaxation
inspired by the Sherali-Adams hierarchy. 

To the best of our knowledge, no approximation algorithms are known for arbitrary substrate graphs and classes of virtual networks containing cyclic substructures.

\paragraph*{Bibliographic Note}

This technical report is an extended version of our paper presented at IFIP Networking 2018~\cite{rostSchmidVNEP_RR_IFIP_18}.

In our preliminary technical report~\cite{DBLP:journals/corr/RostS16} similar results were presented. The current work presents a significantly simpler LP formulation and also provides an extensive computational evaluation. 
Additionally, in our recent technical report~\cite{rostSchmidFPTApproximationsTechReport}, the approximation approach presented in this work is extended beyond cactus request graphs. However, approximating more general request graphs comes at the price of non-polynomial runtimes.

\subsection{Outline of Randomized Rounding for the VNEP}
\label{sub:outline-of-rr-for-the-vnep}

We shortly revisit the concept of randomized rounding~\cite{Raghavan-Thompson}. Given an Integer Program for a certain problem, randomized rounding works by (i) computing a solution to its Linear Program relaxation, (ii) decomposing this solution into convex combinations of \emph{elementary} solutions, and (iii)~probabilistically selecting elementary solutions based on their weight. 

Accordingly, for applying randomized rounding for the VNEP, a convex combination of valid mappings \mbox{$\PotEmbeddings = \{ (\prob, \mapping) | \mapping \in \spaceSolReq, \prob > 0\}$} must be recovered from the Linear Programming solution for each request $\req \in \requests$, such that (i) the profit of these convex combinations equals the profit achieved by the Linear Program and (ii) the (fractional) cumulative allocations do not violate substrate capacities. \emph{Rounding}  a solution is then done as follows. For each request $\req \in \requests$, the mapping $\mapping$ is selected with probability $\prob$, rejecting $\req$ with probability $1- \sum_k \prob$.

\subsection{Results and Organization}

This paper initiates the study of approximation algorithms for the VNEP on general substrates \emph{and} general virtual networks. Specifically, we employ randomized rounding to obtain the first approximation algorithm for the non-trivial class of cactus graph requests in the resource augmentation model. 

Studying the classic multi-commodity flow (MCF) formulation for the VNEP in Section~\ref{sec:classic-mcf-and-its-limits}, we show that its solutions can only be decomposed for tree requests: request graphs containing cycles can in general not be decomposed into valid mappings. This result has ramifications beyond the inability to apply randomized rounding: we prove that the MCF formulation exhibits an unbounded integrality gap.
Investigating the root cause for this surprising result, we devise a novel \emph{decomposable} Linear Programming formulation in Section~\ref{sec:decomposable-novel-formulation} for the class of cactus graph requests.
We then present our randomized rounding algorithm and prove its performance guarantees in Section~\ref{sec:approximation-via-randround}, obtaining the first approximation algorithm for the Virtual Network Embedding Problem.
Section~\ref{sec:evaluation} presents a synthetic computational study, in which two types of rounding heuristics are evaluated. Our results indicate that high-quality solutions can be obtained even without resource augmentations. Indeed, our heuristic rounding algorithm achieves $73.8\%$ of the baseline's profit on average, while not exceeding resource capacities.

\begin{figure}[b!]
\vspace{-6pt}

 {
  \LinesNotNumbered
  \renewcommand{\arraystretch}{1.2}
 
 \removelatexerror

  \begin{IPFormulationStar}{H}
 
  \SetAlgorithmName{Formulation}{}{{}}

  \newcommand{\spaceIt}{\qquad\quad\quad}
  \newcommand{\miniSpace}{\hspace{1.5pt}}
 
\popline
\scalebox{0.95}{
  \begin{tabular}{FRLQF}
   \multicolumn{4}{C}{\textnormal{max~}  \sum \limits_{\req \in \requests}  \Vprofit x_{\req}~~~~~~~~~~~~~~~~  } & \tagIt{alg:VNEP-old:obj}\\
   \sum \limits_{u \in \VVloc} \hspace{-4pt} y^u_{\req, i} & = & x_{\req} & \forall \req \in \requests, i \in \VV &  \tagIt{alg:VNEP-old:node-embedding} \\
      \sum \limits_{u \in \SV \setminus \VVloc} \hspace{-12pt} y^u_{\req, i} & = & 0 & \forall \req \in \requests, i \in \VV &  \tagIt{alg:VNEP-old:node-embedding-non-mappable} \\
   \left[ \begin{array}{l}
      \hspace{-6pt} \sum \limits_{(u,v) \in  \delta^+(u)} \hspace{-14pt} z^{u,v}_{\req,i,j} \hspace{2pt} \\
      \hspace{8pt} - \hspace{-10pt} \sum \limits_{(v,u) \in  \delta^-(u)} \hspace{-14pt} z^{v,u}_{\req,i,j}
      \end{array} \hspace{-6pt} \right] & = & \left[ \hspace{-6pt} \begin{array}{l}
      y^u_{\req, i} \\
      - y^u_{\req,j}
      \end{array} \hspace{-6pt}\right]  & \forall \left[ \begin{array}{l}
       \req \in \requests, (i,j) \in  \VE, \\
       u \in \SV
      \end{array} \right] &  \tagIt{alg:VNEP-old:edge-embedding}\\
   
      z^{u,v}_{\req,i,j} & = & 0 & \forall \left[ \begin{array}{l}
              \req \in \requests, (i,j) \in  \VE, \\
             (u,v) \in \SE \setminus \VEloc
            \end{array} \right]~ &  \tagIt{alg:VNEP-old:edge-embedding-non-mappable}\\

 	\sum \limits_{i \in \VV, \Vtype(i) = \type} \hspace{-16pt}  \Vcap(i) \cdot y^u_{\req,i}  & =  & a^{\type,u}_{\req} & \forall \req \in \requests, (\type,u) \in  \SRV &  \tagIt{alg:VNEP-old:load-node}\\
 	
	\sum \limits_{(i,j) \in  \VE }   \hspace{-8pt} \Vcap(i,j) \cdot z^{u,v}_{\req,i,j} & = &  a^{u,v}_{\req} & \forall \req \in \requests, (u,v) \in  \SE & \tagIt{alg:VNEP-old:load-edge} \\
 	
 	\sum \limits_{\req \in \requests} a^{x,y}_{\req}  & \leq  & \Scap(x,y) & \forall (x,y) \in  \SR &  \tagIt{alg:VNEP-old:capacities} 
  \end{tabular}
}
  \caption{Classic MCF Formulation for the VNEP}
  \label{alg:VNEP-IP-old}
  \end{IPFormulationStar}
  }
\end{figure}

\section{The Classic Multi-Commodity Formulation \\for the VNEP and its Limitations}
\label{sec:classic-mcf-and-its-limits}

In this section, we study the relaxation of the standard multi-commodity flow (MCF) formulation for the VNEP (cf. \cite{mehraghdam2014specifying,vnep}). We first show  the positive result that the formulation is sufficiently strong to decompose virtual networks being \emph{trees} into convex combinations of valid mappings. Subsequently, we show that the formulation fails to allow for the decomposition of \emph{cyclic requests}. This not only impacts its applicability for randomized rounding but renders the formulation useless for approximations in general: we show that the formulation's integrality gap is unbounded.

\subsection{The Classic Multi-Commodity Formulation}

The classic MCF formulation for the VNEP is presented as Formulation~\ref{alg:VNEP-IP-old} . We first describe its (mixed-)integer variant, which computes a single valid mapping for each request by using binary variables. The Linear Programming variant is obtained by relaxing the binary variables' domain to $[0,1]$.

The variable $x_{\req} \in \{0,1\}$ indicates whether request $\req \in \requests$ is embedded or not. The variable $y^u_{\req,i} \in \{0,1\}$ indicates whether virtual node $i \in \VV$ was mapped on substrate node $u \in \SV$. Similarly, the flow variable $z^{u,v}_{\req,i,j} \in \{0,1\}$ indicates whether the substrate edge $(u,v) \in \SE$ is used to realize the virtual edge $(i,j) \in \VE$. The variable $a^{x,y}_{\req} \geq 0$ denotes the cumulative allocations of request $\req \in \requests$ induced on resource $(x,y) \in \SR$.

By Constraint~\ref{alg:VNEP-old:node-embedding}, the virtual node $i \in \VV$ of request $\req \in \requests$ must be placed on any of the suitable substrate nodes of $\VVloc$ iff. $x_\req = 1$ holds and Constraint~\ref{alg:VNEP-old:node-embedding-non-mappable} forbids the mapping on nodes which may not host node $i$. Constraint~\ref{alg:VNEP-old:edge-embedding} induces an unsplittable unit flow for each virtual edge $(i,j) \in  \VE$ from the substrate location to which $i$ was mapped to the substrate location to which $j$ was mapped. By Constraint~\ref{alg:VNEP-old:edge-embedding-non-mappable} virtual edges may only be mapped on \emph{allowed} substrate edges.
Constraints~\ref{alg:VNEP-old:load-node} and \ref{alg:VNEP-old:load-edge} compute the cumulative allocations and Constraint~\ref{alg:VNEP-old:capacities} guarantees that the substrate resource capacities are respected. The following lemma states the connectivity property enforced by Formulation~\ref{alg:VNEP-IP-old}.

\begin{lemma}[Local Connectivity Property of Formulation~\ref{alg:VNEP-IP-old}]~\\
\label{lem:local-connectivity-property}
For any virtual edge $(i,j) \in \VE$ and any substrate node $u \in \VVloc$ with $y^u_{\req,i} > 0$, there exists a path $P^{u,v}_{\req,i,j}$ in $\SG$ from $u$ to $v \in \VVloc[j]$ with $y^{v}_{\req,j} > 0$, such that the flow along any edge of $P^{u,v}_{\req,i,j}$ with respect to the variables $z^{\cdot,\cdot}_{\req,i,j}$ is greater 0. 

The path $P^{u,v}_{\req,i,j}$ can be computed in polynomial time.
\end{lemma}
\begin{proof}
First, note that $\sum_{u \in \VVloc} y^u_{\req,i} = \sum_{v \in \VVloc[j]} y^v_{\req,i}$ holds by Constraint~\ref{alg:VNEP-old:node-embedding} and hence both virtual nodes $i$ and $j$ must be mapped to an equal extent on suitable substrate nodes. Fix any substrate node $u \in \VVloc$ with $y^u_{\req,i} > 0$. 
 If $j$ is also partially mapped on $u$, i.e. if $y^u_{\req,j} > 0$ holds, then the result follows directly, as $u$ connects to $u$ using (and allowing) the empty path $P^{u,u}_{\req,i,j} = \langle \rangle$.
 If, on the other hand, $y^u_{\req,j} = 0$ holds, then Constraint~\ref{alg:VNEP-old:edge-embedding} induces a flow of value $y^u_{\req,i}$ at $u$ in the commodity $z^{\cdot,\cdot}_{\req, i,j}$. As the right hand side of Constraint~\ref{alg:VNEP-old:edge-embedding} may only attain negative  values at nodes $v \in \VVloc[j]$ for which $y^v_{\req,j} > 0$ holds, the flow emitted at node $u$ must eventually reach a node $v \in \VVloc[j]$ with $y^v_{\req,j} > 0$ and hence the result follows.
 The path $P^{u,v}_{\req,i,j}$ can be constructed in polynomial-time by a simple breadth-first search that only considers edges $(u',v') \in \SE$ for which $z^{u',v'}_{\req,i,j} > 0$ holds.
 \end{proof}

\subsection{Decomposing Solutions for Tree Requests}
Given Lemma~\ref{lem:local-connectivity-property}, we now present Algorithm~\ref{alg:decompositionAlgorithm-MCF-Tree} to decompose solutions to the LP Formulation~\ref{alg:VNEP-IP-old} into convex combinations of valid mappings $\PotEmbeddings = \{ (\prob, \mapping) | \mapping \in \spaceSolReq, \prob > 0\}$~(cf.~Section~\ref{sub:outline-of-rr-for-the-vnep}), \emph{if} the request's underlying undirected graph is a \emph{tree}. Recall that in the LP formulation the binary variables are relaxed to take any value in the interval $[0,1]$. 

Given a request $\req \in \requests$, the algorithm processes all virtual edges according to an arbitrary acyclic representation $\VGbfs = (\VV,\VEbfs,\VVroot)$ of the undirected interpretation of $\VG$ being rooted at $\VVroot \in \VV$. Concretely, the edge set $\VEbfs$ is obtained from $\VE$ by reorienting (some of the) edges, such that any node $i \in \VV$ can be reached from $\VVroot$.  Considering tree requests for now, $\VGbfs$ is an arborescence and can be computed by a simple graph search of the underlying undirected graph starting at an arbitrary root node $\VVroot \in \VV$. We denote by $\VEDiff = \VE \setminus \VEbfs$ the edges whose orientations were reversed in the process of computing~$\VGbfs$.

\begin{figure}[tb]

\scalebox{0.85}{
\begin{minipage}{1.16\columnwidth}

\removelatexerror

\begin{algorithm*}[H]

\SetKwInOut{Input}{Input}\SetKwInOut{Output}{Output}
\SetKwFunction{ProcessPath}{ProcessPath}{}{}
\SetKwFunction{reverse}{reverse}{}{}
\SetKwFunction{LP}{LP}
\SetKwFunction{LP}{LP}

\newcommand{\SET}{\textbf{set~}}
\newcommand{\ADD}{\textbf{add~}}
\newcommand{\DEFINE}{\textbf{define~}}
\newcommand{\AND}{\textbf{and~}}
\newcommand{\LET}{\textbf{let~}}
\newcommand{\WITH}{\textbf{with~}}
\newcommand{\COMPUTE}{\textbf{compute~}}
\newcommand{\CHOOSE}{\textbf{choose~}}
\newcommand{\DECOMPOSE}{\textbf{decompose~}}
\newcommand{\FORALL}{\textbf{for all~}}
\newcommand{\OBTAIN}{\textbf{obtain~}}
\newcommand{\WITHPROBABILITY}{\textbf{with probability~}}

\Input{Tree request $\req \in \requests$, solution~$(x_{\req},\vec{y}_{\req},\vec{z}_{\req},\vec{a}_{\req})$ to LP Formulation~\ref{alg:VNEP-IP-old}, acyclic reorientation $\VGbfs = (\VV,\VEbfs, \VVroot)$}
\Output{Convex combination~$\PotEmbeddings = \{(\prob,\mapping)\}_k$}
\BlankLine

	\SET $\PotEmbeddings  \gets \emptyset$ \AND $k \gets 1$\\
	\While{$x_{\req} > 0$ }
	{
		
		\SET $\mapping = (\mapV,\mapE)~\gets (\emptyset,\emptyset)$ \label{alg:sc-decomposition:init-maps}\\
		
		\SET $\Queue = \{\VVroot \}$\\
		
		\CHOOSE $u \in \VVloc[\VVroot]$ \WITH $y^u_{\req,\VVroot} > 0$ \AND \SET $\mapV(\VVroot)~\gets u$\\
		
		\While{$|\Queue| > 0$}{	\label{alg:decomposition:begin-while-q}
			\CHOOSE $i \in \Queue$ \AND \SET$\Queue \gets \Queue \setminus \{i\}$\\
			\ForEach{$(i,j) \in \VEbfs$}{
				\eIf{$(i,j) \in \VE$}{
					\COMPUTE $\overrightarrow{P}^{u,v}_{\req,i,j}$ connecting $\mapV(i)=u$ to $v \in \VVloc[j]$ \label{alg:decomp:compute-path-original-orientation}\\
					\pushline\pushline~~~\,\nonl according to Lemma~\ref{lem:local-connectivity-property}\\
					\popline\popline
					\SET $\mapV(j) = v$ \AND $\mapE(i.j) = \overrightarrow{P}^{u,v}_{\req,i,j}$\\
				}{
					\LET $\overleftarrow{z}^{v',u'}_{\req,i,j} \triangleq z^{u',v'}_{\req,j,i}$ \FORALL $(u',v') \in \SE$ \label{alg:decomp:reverse-start}\\
					\COMPUTE $\overleftarrow{P}^{v,u}_{\req,i,j}$ connecting $\mapV(i)=v$ to $u \in \VVloc[j]$ \label{alg:decomp:compute-path-reverse-orientation}\\ 
					\pushline\pushline~~~\,\nonl according to Lemma~\ref{lem:local-connectivity-property}\\
					\popline\popline
					\SET $\overrightarrow{P}^{u,v}_{\req,j,i} = \reverse(\overleftarrow{P}^{v,u}_{\req,i,j})$\\
					\SET $\mapV(i) = u$ \AND  $\mapE(j,i) = \overrightarrow{P}^{u,v}_{\req,j,i} $ \label{alg:decomp:reverse-end}\\

				}
				\SET $\mathcal{Q} \gets \mathcal{Q} \cup \{j\}$\\
			}
		}
		\SET $\mathcal{V}_k \gets \left( \hspace{-6pt}  \begin{array}{rcl}
		\{x_\req\} \hspace{-4pt} & \hspace{-4pt}  \cup \hspace{-4pt}  & \hspace{-4pt}  \{y^{\mapV(i)}_{\req,i} | i \in \VV\}  \\
		&			\hspace{-4pt}  \cup  \hspace{-4pt}  & \hspace{-4pt}  \{z^{u,v}_{\req,i,j} | (i,j) \in \VE, (u,v) \in \mapE(i,j)\}
		\end{array} \hspace{-6pt}  \right)$ \label{alg:decomposition:compute-Vk}\\
		\SET $\prob \gets \min \mathcal{V}_k$ \label{alg:decomposition:computing-prob} \\
		\SET $v \gets v - \prob$ \FORALL $v \in \mathcal{V}_k$ \label{alg:decomposition:reduce-flow-variables}\\
		\SET $a^{x,y}_{\req} \gets a^{x,y}_{\req} - \prob \cdot A(\mapping,x,y)$ \FORALL $(x,y) \in \SR$ \label{alg:decomposition:adapt-variables}\\	
		\ADD $D^k_{\req} = (f^k_{\req},m^k_{\req})$ to $\PotEmbeddings$ \AND \SET $k \gets k + 1$\\
	}

\KwRet{$\PotEmbeddings$}
\caption{Decompositioning MCF solutions for Tree Requests}
\label{alg:decompositionAlgorithm-MCF-Tree}
\end{algorithm*}
\end{minipage}}
\end{figure}

The algorithm extracts mappings $\mapping$ of value $\prob$ iteratively, as long as $x_{\req} > 0$ holds. Initially, in the $k$-th iteration, none of the virtual nodes and edges are mapped. As $x_\req > 0$ holds, there must exist a node $u \in \VVloc[\VVroot]$ with $y^{u}_{\req,\VVroot} > 0$ by Constraint~\ref{alg:VNEP-old:node-embedding}  and the algorithm accordingly sets $\mapV(\VVroot) = u$.  Given this initial fixing, the algorithm iteratively extracts nodes from the queue $\mathcal{Q}$ which have been already mapped and considers all outgoing  virtual edges $(i,j) \in \VEbfs$. If an outgoing edge $(i,j)$ is contained in $\VE$, Lemma~\ref{lem:local-connectivity-property} can be readily applied to obtain a joint mapping of the edge $(i,j)$ and its head $j$.
If the edge's orientation was reversed, i.e.~if $(i,j) \in \VEDiff$ holds,  Lemma~\ref{lem:local-connectivity-property} is applied \emph{while reversing} the flow's direction (Lines \mbox{\ref{alg:decomp:reverse-start}~-~\ref{alg:decomp:reverse-end}}).

First, note that by the repeated application of Lemma~\ref{lem:local-connectivity-property}, the mapping of virtual nodes and edges is valid. As $\VGbfs$ is an arborescence, each edge and each node of $\VGbfs$ will eventually be mapped and hence $\mapping$ is a valid mapping.
The mapping value $\prob$ is computed as the minimum of the mapping variables $\mathcal{V}_k$ used for constructing $\mapping$. Reducing the values of the mapping variables together with the allocation variables $\vec{a}_{\req}$ (Lines~\ref{alg:decomposition:reduce-flow-variables} and \ref{alg:decomposition:adapt-variables}), the Constraints~\ref{alg:VNEP-old:node-embedding}~-~\ref{alg:VNEP-old:load-edge} continue to hold.

As the decomposition process continues as long as $x_{\req} > 0$ holds and in the $k$-th iteration at least one variable's value is set to $0$, the algorithm terminates with a complete decomposition for which $\sum_{k} \prob = x_{\req}$ holds. Furthermore, the algorithm has polynomial runtime, as in each iteration at least one variable is set to 0 and the number of variables for request $\req$ is bounded by $\mathcal{O}(|\VG| \cdot |\SG|)$. Hence, we obtain the following:

\begin{lemma}
\label{lem:decomposability-mcf-trees}
Given a virtual network request $\req \in \requests$, whose underlying undirected graph is a tree, Algorithm~\ref{alg:decompositionAlgorithm-MCF-Tree} decomposes a solution $(x_\req, \vec{y}_{\req}, \vec{z}_{\req}, \vec{a}_{\req})$ to LP Formulation~\ref{alg:VNEP-IP-old} into valid a convex combination of valid mappings \mbox{$\PotEmbeddings = \{(\mapping,\prob) | \mapping \in \spaceSolReq, \prob > 0\}_k$}, such that:
\begin{itemize}
\item The decomposition is complete, i.e. $x_{\req} = \sum_k \prob$ holds.
\item The decomposition's allocations are bounded by $\vec{a}_{\req}$, i.e. \newline $a^{x,y}_{\req} \geq  \sum_{k} \prob \cdot A(\mapping,x,y)$ holds for $(x,y) \in \SR$.
\end{itemize}
\end{lemma}

\subsection{Limitations of the Classic MCF Formulation}

Above it was shown that LP solutions to the classic MCF formulation can be decomposed into convex combinations of valid mappings \emph{if} the underlying graph is a tree. This does not hold anymore when considering cyclic virtual networks:

\begin{figure}[t!]
\centering
\includegraphics[width=1.0\columnwidth]{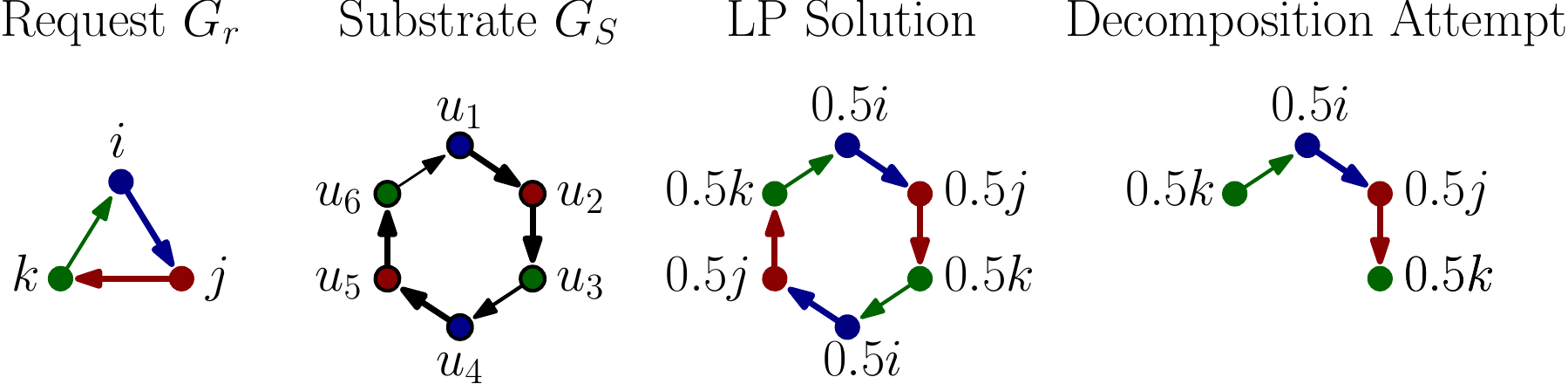}
\caption{Example showing that solutions to the LP Formulation~\ref{alg:VNEP-IP-old} can in general not be decomposed into convex combinations of valid mappings. Request $\req$ is a simple cyclic graph which shall be mapped on the substrate graph $\SG$. We assume following node mapping restrictions $\VVloc[i] = \{u_1,u_4\}$, $\VVloc[j] = \{u_2,u_5\}$, $\VVloc[k] = \{u_3,u_6\}$.
The LP solution with $x_\req=1$ is depicted as follows. Substrate nodes are annotated with the mapping of virtual nodes. Hence, $0.5i$ at node $u_1$ indicates $y^{u_1}_{r,i} = 1/2$, i.e. that virtual node $i$ is mapped with $0.5$ on substrate node $u_1$. Substrate edges are colored according to the color of virtual links mapped onto it. Virtual links are all mapped using flow values $1/2$. Accordingly, for example $z^{u_1,u_2}_{r,i,j} = 1/2$ holds. }
\label{fig:non-decomp}
\end{figure}

\begin{theorem}
Solutions to the standard LP Formulation~\ref{alg:VNEP-IP-old} can in general not be decomposed into convex combinations of valid mappings if the virtual networks contain cycles.
\label{lem:non-decomposability}
\end{theorem}
\begin{proof}
In Figure~\ref{fig:non-decomp} we visually depict an example of a solution to the LP Formulation~\ref{alg:VNEP-IP-old} from which \emph{not a single} valid mapping can be extracted. 
The validity of the depicted solution follows from the fact that the virtual node mappings sum to $1$ and each virtual node connects to its neighboring node with half a unit of flow.
Assume for the sake of contradiction that the depicted solution can be decomposed. As virtual node $i \in \VV$ is mapped onto substrate node $u_1 \in \SV$, and $u_2 \in \SV$ is the only neighboring node with respect to variables  $z^{\cdot,\cdot}_{\req, i,j}$ that hosts $j \in \VV$, there must exist a mapping $(\mapV, \mapE)$ with $\mapV(i)=u_1$ and $\mapV(j)=u_2$. Similarly, $\mapV(k)=u_3$ must hold. However, for $\mapV(i) = u_1$, the virtual node $k$ must be mapped to $u_6$, as otherwise the embedding of $(k,i)$ cannot lead to substrate node $u_1$.
Hence the virtual node $k \in \VV$ must be mapped both on $u_6$ and $u_3$. As this is not possible, and the same argument holds when considering the mapping of $i$ onto $u_4$, no valid mapping can be extracted from the LP solution.
\end{proof}

This non-decomposability also induces large integrality gaps. Specifically, when considering edge mapping restrictions the integrality gap of the MCF Formulation~\ref{alg:VNEP-IP-old} is unbounded.
\begin{restatable}{theorem}{integralityGapInfinite}
The integrality gap of the MCF formulation is unbounded. This even holds under infinite substrate capacities.
\label{lem:non-decomposability-integrality-gap}
\end{restatable}

\begin{proof}
\vspace{-12pt}
We again consider the example of Figure~\ref{fig:non-decomp}. We introduce the following restrictions for mapping the virtual links: $\VEloc[i,j] = \{(u_1,u_2), (u_4,u_5)\}$, $\VEloc[j,k] = \{(u_2,u_3), (u_5,u_6)\}$, $\VEloc[k,i] = \{(u_3,u_4), (u_6,u_1)\}$.
Note that the LP solution depicted in Figure~\ref{fig:non-decomp} is still feasible and hence the LP will attain an objective of $\Vprofit$.

On the other hand, there does not exist a valid mapping of request $\req$: assume $i$ to be mapped on $u_1$, then $j$ must be mapped on $u_2$ and $k$ must be mapped on $u_3$ due to the node \emph{and} edge mapping. However, the edge mapping restrictions for $(k,i)$ do not allow the establishment of a path from $u_3$ to $u_1$. By the same argument, the mapping of $i$ on $u_4$ is not feasible.
Accordingly, as the optimal solution achieves a profit of $0$ while the LP solution yields a profit of $\Vprofit$, the integrality gap of the MCF LP Formulation~\ref{alg:VNEP-IP-old} is unbounded.
\end{proof}

As the above proof strongly relied on the fact that edge mapping restrictions are employed, we also derive the following result when only node mapping restrictions are considered.

\begin{restatable}{theorem}{integralityGapOmega}
The integrality gap of the MCF formulation lies in $\Omega(|\SV|)$, when only considering node mapping restrictions.
\label{lem:non-decomposability-integrality-gap-only-node}
\end{restatable}
\begin{proof}
\vspace{-12pt}
Consider the following instance. The substrate is a cycle of an even number of $n$ nodes $u_i$, with $1 \leq i \leq n$ and edges $\{(u_1,u_2), (u_2,u_3), \dots, (u_{n-1},u_n), (u_n,u_1)\}$. We consider unit edge capacities, i.e. we set $\Scap(e) = 1$ for $e \in \SE$. 
Consider now a request $\req$ with $\VV = \{i,j\}$ and edges $\VE \{(i,j), (j,i)\}$ and unit bandwidth demands $\Vcap(i,j) = \Vcap(j,i) = 1$. Assume that $i$ can only be mapped on substrate nodes having an uneven index and that $j$ can only be mapped on substrate nodes of even index.
Clearly, any valid mapping of this request will use \emph{all} edge resources. Consider now the following MCF solution with $x_{\req} = 1$: nodes $i$ and $j$ are, together with the respective edges, mapped in an alternating fashion: $y^{u}_{\req,i} = y^{v}_{\req,j} = 2/n$ holds for all $u \in \VVloc[i] = \{u_1,u_3,\dots\}$ and all $v \in \VVloc[j] = \{u_2,u_4,\dots\}$, respectively. Similarly, $z^{e}_{\req,i,j} = z^{e'}_{\req,j,i} = 2/n$ is set for all edges $e$ and $e'$ originating at uneven and even nodes, respectively. Hence, the allocation equals $2/n$  on each edge. Hence, $n/2 \in \Theta(|\SV|)$ many copies of this request may be embedded in the MCF solution, while the optimal solution may only embed a single request and the integrality gap therefore lies in $\Omega(|\SV|)$.
\end{proof}

\section{Novel Decomposable LP Formulation}
\label{sec:decomposable-novel-formulation}

In this section, we present a novel LP formulation and its accompanying decomposition algorithm for the class of cactus request graphs, i.e. graphs for which cycles intersect in at most a single node (in its undirected interpretation). Accordingly, these graphs can be uniquely decomposed into cycles and a single forest (cf. Lemma~\ref{lem:observations-bfs-graphs} below). 

Before delving into the details of our novel LP formulation, we discuss our main insight on how to overcome the limitations of the MCF formulation and accordingly how to derive decomposable formulations. To this end, it is instructive, to revisit the non-decomposable example of Figure~\ref{fig:non-decomp} by applying the decomposition Algorithm~\ref{alg:decompositionAlgorithm-MCF-Tree} on the depicted LP solution. Concretely, we consider the acyclic reorientation $\VGbfs = (\VV, \VEbfs,\VVroot)$ with $\VEbfs = \{(i,k), (i,j), (j,k)\}$, such that $i$ is the root, $\VVroot = i$. Assuming that $i$ is initially mapped on node $u_1$, Algorithm~\ref{alg:decompositionAlgorithm-MCF-Tree} will map edges $(i,k)$ and $(i,j)$ first, setting $\mapV(k)=u_6$ and $\mapV(j) = u_2$ However, when the edge $(j,k)$ is processed, $k$ must be mapped on substrate node $u_3 \neq \mapV(k)$ and the algorithm hence fails to produce a valid mapping. Accordingly, to avoid such \emph{diverging} node mappings, our key idea is to decide the mapping location of nodes with more than one incoming edge (with respect to the request's acyclic reorientation $\VGbfs$) \emph{a priori}. 

By considering only cactus request graphs, the above can be implemented rather easily as exactly one node of each cycle has more than one incoming edge: one only needs to ensure \emph{compatibility} of node mappings for this node. To resolve potential conflicts for the mapping of this unique \emph{cycle target}, our formulation employs \emph{multiple} copies of the MCF formulation for the respective cyclic subgraph. Specifically, considering a cycle with virtual \emph{target} node $k$, we instantiate one MCF formulation per substrate node $w \in \VVloc[k]$ onto which $k$ can be mapped. Accordingly, this yields at most $|\SV|$ many copies and for each of these copies $k$ is \emph{fixed} to one specific (substrate) mapping location. Accordingly, as the mapping location of $k$ is fixed to a specific node, valid mappings for the respective cycles can always be extracted from such a MCF copy: the mappings of $k$ cannot possibly diverge.

\subsection{Cactus Request Graph Decomposition and Notation}

Based on the assumed cactus graph nature of request graphs, we consider the following decomposition of request graphs.

\begin{lemma} 
\label{lem:observations-bfs-graphs}
Consider a cactus request graph $\VG$ and an acyclic reorientation $\VGbfs$ of $\VG$.
The graph $\VGbfs$ can be uniquely partitioned into subgraphs $\{ \VGcycle[\req][1], \ldots, \VGcycle[\req][n] \} \sqcup \VGforest$, such that the following holds:
\begin{enumerate}
\item The subgraphs $\{ \VGcycle[\req][1], \ldots, \VGcycle[\req][n] \}$ correspond to the (undirected) cycles of $\VG$ and $\VGforest$ is the \emph{forest} remaining after removing the cyclic subgraphs. We denote the index set of the cycles by $\Cycles = \{C_1, \ldots,C_n\}$.
\item The subgraphs partition the edges of $\VEbfs$: an edge \mbox{$(i,j) \in \VEbfs$} is contained in exactly one of the subgraphs.
\item The edge set  $\VEcycle$  of each cycle $\cycle \in \Cycles$ can itself be partitioned into two branches $\mathcal{B}^{\cycle}_{1}$ and $\mathcal{B}^{\cycle}_{2}$, such that both lead from $\VVcycleSource \in \VVcycle$ to  $\VVcycleTarget \in \VVcycle$.
\end{enumerate}
\end{lemma}
\begin{proof}
We prove that the lemma holds independently of the chosen acyclic reorientation. Accordingly, let $\VGbfs$ be an arbitrary acyclic reorientation of $\VG$. We first show that $|\delta^-(i)| \leq 2$ holds for each virtual node $i \in \VV$ with respect to the edge set $\VEbfs$, i.e. any virtual node has at most $2$ incoming edges in $\VGbfs$. If $|\delta^-(i)| > 2$ held, then by the definition of the acyclic reorientation there must exist at least $3$ paths $P_1$, $P_2$, $P_3$ from the root $\VVroot$ to $i$. Let $p_{1,2}, p_{2,3} \in \VV$ be the last common nodes lying on both $P_1$ and $P_2$ and $P_2$ and $P_3$, respectively. Clearly, from $p_{1,2}$, there exist two (otherwise) node-disjoint paths to $i$ and from $p_{2,3}$ there exist two (otherwise) node-disjoint paths towards $i$. Accordingly, there exist at least two cycles intersecting either in $i$ and $p_{1,2}$ or $i$ and $p_{2,3}$, which contradicts the assumption on $\VG$ being a cactus graph. Accordingly, $|\delta^-(i)| \leq 2$ holds for all virtual nodes $i \in \VV$ according to $\VEbfs$.

Now, consider any node $i \in \VV$ with $|\delta^-(i)| \leq 2$. By performing a graph-search in the opposite direction of edges in $\VEbfs$, a unique common ancestor node $i'$ can be determined, such that there exist exactly two paths $\mathcal{B}_1$ and $\mathcal{B}_2$ (branches) from $i'$ to $i$. The union of these branches represents a single `cycle' (cf. Statement~3). Removing the identified cycle from the graph, the cactus graph property still holds, as removing edges can never refute it. Accordingly, `cycles' in the acyclic reorientation $\VGcycle[\req][1], \ldots, \VGcycle[\req][n]$ can be uniquely identified (decomposed) and after repeated removal only the forest $\VGforest$ remains. Hence, the first statement of the lemma follows. Lastly, the second statement of the lemma holds trivially as each edge is either contained in any of the cycles or is part of the remaining forest.
\end{proof}

Besides the above introduced notation, we denote by $\VGcycleOrig$ and $\VGforestOrig$ the subgraphs that agree with $\VG$ on the edge orientations and use $\SVTypesCycle=\VVloc[\VVcycleTarget]$ to denote the substrate nodes on which $\VVcycleTarget$ can be mapped.

\subsection{Novel LP Formulation for Cactus Requests}

Our novel Formulation~\ref{IP:novel} uses the a priori partition of $\VGbfs$ into cycles $\VGcycle$ and the forest $\VGforest$ to construct MCF formulations for the respective subgraphs: for the subgraph $\VGforestOrig$ a single copy is used (cf. Constraint~\ref{LP:novel:forest}) while for the cyclic subgraphs a single MCF formulation is employed \emph{per} potential target location contained in the set $\SVTypesCycle$ (cf. Constraint~\ref{LP:novel:cycles}).
We index the variables of these sub-LPs by employing square brackets.

To bind together these (at first) independent MCF formulations, we reuse the variables  $\vec{x},\vec{y}$, and $\vec{a}$ introduced already for the MCF formulation. We refer to these variables, which are defined outside of the sub-LP formulations, as \emph{global variables} and do not index these. As we only consider the LP formulation, all variables are continuous.

The different sub-formulations are linked as follows. We employ Constraint~\ref{LP:novel:node-embedding} to enforce the setting of the (global) node mapping variables (cf. Constraint~\ref{alg:VNEP-old:node-embedding} of Formulation~\ref{alg:VNEP-IP-old}). By Constraints~\ref{LP:novel:node-mapping-forest} and \ref{LP:novel:node-mapping-cycles}, the node mappings of the sub-LPs for mapping the subgraphs must agree with the global node mapping variables. With respect to cyclic subgraphs, we note that Constraint~\ref{LP:novel:node-mapping-cycles} allows for distributing the global node mappings to any of the $|\SVTypesCycle|$ formulations: only the sum of the node mapping variables must agree with the global node mapping variable. Constraint~\ref{LP:novel:node-mapping-cycles-restriction} is of crucial importance for the decomposability: considering the sub-LP for cycle $\cycle$ and target node $w \in \SVTypesCycle$, it enforces that the target node $\VVcycleTarget$ of the cycle $\cycle$ \emph{must} be mapped on $w$. Thus, in the sub-LP $[\cycle, w]$ both branches $\mathcal{B}^{\cycle}_1$ and $\mathcal{B}^{\cycle}_2$ of cycle $\cycle$ are \emph{pre-determined} to lead to the node $w$.
Lastly, for computing node allocations the global node mapping variables are used (cf. Constraint~\ref{LP:novel:node-load}) and for computing edge allocations the sub-LP formulations' allocations are considered (cf. Constraint~\ref{LP:novel:edge-load}).

\begin{figure}[b!]

 {
  \LinesNotNumbered
  \renewcommand{\arraystretch}{1.4}
 
 \removelatexerror

  \begin{IPFormulation}{H}

  \SetAlgorithmName{Formulation}{}{{}}

  \newcommand{\spaceIt}{\qquad\quad\quad}
  \newcommand{\miniSpace}{\hspace{1.5pt}}
 
  \popline
\scalebox{0.85}{
\begin{minipage}{1.25\columnwidth}
  \begin{tabular}{FRLQB}
   \multicolumn{4}{C}{\textnormal{max~}  \sum \limits_{\req \in \requests}  \Vprofit x_{\req}~~~~~~~~~~~~~~~~  } & \tagIt{LP:novel:obj}\\

   \multicolumn{3}{L}{   
   \begin{array}{l}
\textnormal{Cons. (\ref{alg:VNEP-old:node-embedding}) - (\ref{alg:VNEP-old:load-edge}) for $\VGforestOrig$ on } \\[-2pt]
   \textnormal{variables }(x_{\req}, \vec{y}_{\req},\vec{z}_{\req}, \vec{a}_{\req})[\forest]
   \end{array}
   }  
   & \forall \req \in \requests    & \tagIt{LP:novel:forest}
   \\
       \multicolumn{3}{L}{ 
       \begin{array}{l}
    \textnormal{Cons. (\ref{alg:VNEP-old:node-embedding}) - (\ref{alg:VNEP-old:load-edge}) for $\VGcycleOrig$ on} \\[-2pt]
      \textnormal{variables }(x_{\req}, \vec{y}_{\req},\vec{z}_{\req}, \vec{a}_{\req})[\cycle,w]
       \end{array}
       }  
    
    &  \forall \req \in \requests, \cycle \in \Cycles,  w \in \SVTypesCycle & \tagIt{LP:novel:cycles}
   \\[10pt]
   
      x_{\req} & = & \sum \limits_{u \in \VVloc} y^u_{\req, i} & \forall \req \in \requests, i \in \VV &  \tagIt{LP:novel:node-embedding} \\[16pt]
   
   	y^u_{\req, i} & = & y^u_{\req,i}[\mathcal{F}] & \forall \req \in \requests, i \in \VVforestOrig, u \in \VVloc &  \tagIt{LP:novel:node-mapping-forest} \\[6pt]
   	
	y^u_{\req,i} & = & \sum_{w \in \VVcycleTarget} y^u_{\req,i}[\cycle,w] & \forall \left[ \hspace{-3pt} \begin{array}{l}
	\req \in \requests,  i \in \VV, u \in \VVloc[i]\hspace{-3pt},\,\,\\[-2pt]
	  \cycle \in \Cycles : i \in \VVcycleOrig
	\end{array}   \hspace{-3pt}  \right] &  \tagIt{LP:novel:node-mapping-cycles} \\
	
	0 & = & y^u_{\req,\VVcycleTarget}[\cycle,w]  & \forall \left[ \begin{array}{l} \hspace{-3pt}
	\req \in \requests, \cycle \in \Cycles, w \in \SVTypesCycle,\\[-2pt]
	\hspace{-3pt} u \in \SVTypesCycle \setminus \{w\}
	\end{array} \hspace{-3pt} \right] &  \tagIt{LP:novel:node-mapping-cycles-restriction} \\[14pt]
	
a^{\type,u}_{\req}  & =  & 	 	\sum \limits_{i \in \VV, \Vtype(i) = \type}  \hspace{-12pt}  \Vcap(i) \cdot y^u_{\req,i} & \forall \req \in \requests, (\type,u) \in  \SRV &  \tagIt{LP:novel:node-load}\\[14pt]
   	
	a^{u,v}_{\req} & = & a^{u,v}_{\req}[\mathcal{F}] + \hspace{-18pt} \sum_{\cycle \in \Cycles, w \in \SVTypesCycle} \hspace{-20pt} a^{u,v}_{\req}[\cycle,w]  ~ & \forall \req \in \requests, (u,v) \in  \SE&  \tagIt{LP:novel:edge-load}\\[16pt]
	\multicolumn{3}{L}{ \sum \limits_{\req \in \requests} a^{x,y}_{\req}  \leq   \Scap(x,y)}
	 & \forall (x,y) \in  \SR &  \tagIt{LP:novel:capacities}  

  \end{tabular}
\end{minipage}}
  \caption{Novel LP for Cactus Requests}
  \label{IP:novel}
  \end{IPFormulation}
  }
\end{figure}

\subsection{Decomposing Solutions to the Novel LP Formulation}

We now show how to adapt the decomposition Algorithm~\ref{alg:decompositionAlgorithm-MCF-Tree} to decompose solutions to Formulation~\ref{IP:novel}.

To decompose the LP solution for a request $\req$, the acyclic reorientation $\VGbfs$, which was also used for constructing the LP, must be handed over to the decomposition algorithm. 

As the novel LP formulation does not contain (global) edge mapping variables, the edge mapping variables used in Lines~\ref{alg:decomp:compute-path-original-orientation} and \ref{alg:decomp:compute-path-reverse-orientation} of Algorithm~\ref{alg:decompositionAlgorithm-MCF-Tree} must be substituted by edge mapping variables of the respective sub-LP formulations. Concretely, as each edge of the request graph $\VG$ is covered exactly once, it is clear whether a virtual edge $(i,j) \in \VE$ is part of $\VGforestOrig$ or a cyclic subgraph $\VGcycleOrig$. If $(i,j) \in \VGforestOrig$ holds, then the edge mapping variables $z^{\cdot,\cdot}_{\req,i,j}[\forest]$ are used. 
If on the other hand the edge $(i,j) \in \VE$ is covered in the cyclic subgraph $\VGcycleOrig$, then there exist $|\SVTypesCycle|$ many sub-LPs to choose the respective edge mapping variables from. To ensure the decomposability, we proceed as follows.

If the edge $(i,j) \in \VEbfs$ is the first edge of $\VGcycleOrig$ to be mapped in the $k$-th iteration, the mapping variables $z^{\cdot,\cdot}_{\req,i,j}[\cycle,w]$  belonging to an arbitrary target node $w$, with $y^{\mapV(i)}_{\req,i}[\cycle,w] > 0$, are used. Such a node $w$ exists by Constraint~\ref{LP:novel:node-mapping-cycles}.

If another edge $(i',j')$ of the same cycle was already mapped in the $k$-th iteration, the \emph{same} sub-LP before is considered. Accordingly, the mapping of cycle target nodes cannot conflict, and as these are the only nodes with potential mapping conflicts, the returned mappings are always valid.

To successfully iterate the extraction process, the steps taken in Lines~\ref{alg:decomposition:compute-Vk}~-~\ref{alg:decomposition:adapt-variables} of Algorithm~\ref{alg:decompositionAlgorithm-MCF-Tree} must be adapted to consider the sub-LP variables. Again, as in each iteration at least a single variable of the LP is set to 0 and as the novel Formulation~\ref{IP:novel} contains at most $\mathcal{O}(|\SV|)$ times more variables than the MCF Formulation~\ref{alg:VNEP-IP-old}, the decomposition algorithm still runs in polynomial-time.
Hence, we conclude that the result of Lemma~\ref{lem:decomposability-mcf-trees} carries over to the novel LP Formulation~\ref{IP:novel} for \emph{cactus} request graphs and state the following theorem.

\begin{theorem}
\label{thm:decomposability-novel-formulation}
Given a solution $(x_\req, \vec{y}_{\req}, \vec{z}_{\req}, \vec{a}_{\req})$ to the novel LP Formulation~\ref{IP:novel} for a cactus request graph $\VG$, the solution can be decomposed into a convex combination of valid mappings \mbox{$\PotEmbeddings = \{(\mapping,\prob) | \mapping \in \spaceSolReq, \prob > 0\}_k$}, such that:
\begin{itemize}
\item The decomposition is complete, i.e. $x_{\req} = \sum_k \prob$ holds.
\item The decomposition's allocations are bounded by $\vec{a}_{\req}$, i.e. \newline $a^{x,y}_{\req} \geq \sum_{k} \prob \cdot A(\mapping,x,y)$ holds for $(x,y) \in \SR$.
\end{itemize}
\end{theorem}

 \begin{figure}[b!]

 \scalebox{0.93}{
  \begin{minipage}{1.05\columnwidth}
  
  \removelatexerror
  
  \begin{algorithm*}[H]

  \SetKwInOut{Input}{Input}\SetKwInOut{Output}{Output}
  \SetKwFunction{ProcessPath}{ProcessPath}{}{}
  \SetKwFunction{reverse}{reverse}{}{}
  \SetKwFunction{LP}{LP}
  \SetKwFunction{LP}{LP}
  
  \newcommand{\SET}{\textbf{set~}}
  \newcommand{\ADD}{\textbf{add~}}
  \newcommand{\DEFINE}{\textbf{define~}}
  \newcommand{\AND}{\textbf{and~}}
  \newcommand{\LET}{\textbf{let~}}
  \newcommand{\WITH}{\textbf{with~}}
  \newcommand{\COMPUTE}{\textbf{compute~}}
   \newcommand{\CHECK}{\textbf{check~}}
     \newcommand{\REMOVE}{\textbf{remove~}}
      \newcommand{\REPEAT}{\textbf{repeat~}}
  \newcommand{\FIND}{\textbf{find~}}
  \newcommand{\CHOOSE}{\textbf{choose~}}
   \newcommand{\CHOOSi}{\textbf{choosing~}}
   \newcommand{\CONS}{{construct solution by~}}
  \newcommand{\DECOMPOSE}{\textbf{decompose~}}
  \newcommand{\FORALL}{\textbf{for all~}}
  \newcommand{\OBTAIN}{\textbf{obtain~}}
  \newcommand{\WITHPROBABILITY}{\textbf{with probability~}}
  
  \SetKwRepeat{Do}{do}{while}%
  
 \ForEach(\tcp*[f]{preprocess requests}){$\req \in \requests$ \label{alg:randround:preprocess-start}}{ 
 	\COMPUTE LP Formulation~\ref{IP:novel} for request $\req$ maximizing $x_r$\\
 	\textbf{if} $x_{\req} < 1$ \textbf{then remove} request $\req$ from the set $\requests$ \label{alg:randround:preprocess-end}\\
 }
 \COMPUTE LP Formulation~\ref{IP:novel} for $\requests$ maximizing $\sum_{\req \in \requests} \Vprofit \cdot x_{\req}$ \label{alg:randround:compute-LP}\\
 \ForEach(\tcp*[f]{perform decomposition}){$\req \in \requests$}{
 \COMPUTE $\PotEmbeddings=\{(\prob,\mapping)\}_k$ from LP solution \label{alg:randround:decompose}\\
 }
 \DontPrintSemicolon
  \Do(\tcp*[f]{perform randomized rounding  \label{alg:randround:rounding-start}}){\label{alg:randround:rounding-end}
 $\left(
\begin{array}{ll}
  \textnormal{solution is \emph{not} $(\alpha,\beta,\gamma)$-approximate and }  \\
  \textnormal{maximal rounding tries are not exceeded}
  \end{array}
 \right)$
 \textnormal{ }}{ 
	 \textbf{foreach} $\req \in \requests$ \textbf{select} $\mapping$ \WITHPROBABILITY $\prob$  \label{alg:randround:rounding-main}\\
 }
  \caption{Randomized Rounding for the VNEP}
  \label{alg:rand-round-profit}
  \end{algorithm*}
 \end{minipage}}
  
  \end{figure}
 
\section{Approximation via Randomized Rounding}
\label{sec:approximation-via-randround}
Above we have shown how convex combinations for the VNEP can be computed for cactus requests. Given these convex combinations, we now present our approximation algorithm (see Algorithm~\ref{alg:rand-round-profit}) and analyze its probabilistic guarantees in Sections~\ref{sec:performance-guarantee-obj-admission-control} to \ref{sec:main-results-admission-control}. Afterwards, we propose several rounding heuristics to apply our approximation in practice (see Section~\ref{sec:main-section-discussion}).

\paragraph*{Synopsis of our Approximation Algorithm}
The algorithm first performs a preprocessing in Lines \ref{alg:randround:preprocess-start}~-~\ref{alg:randround:preprocess-end} by removing all requests which cannot be fully (fractionally) embedded in the absence of other requests, as these can never be part of any feasible solution. In Lines~\ref{alg:randround:compute-LP}~-~\ref{alg:randround:decompose} an optimal solution to the novel LP Formulation~\ref{IP:novel} is computed and afterwards decomposed into convex combinations. Then, in Lines~\ref{alg:randround:rounding-start}~-~\ref{alg:randround:rounding-end}, the rounding is performed: for each request $\req \in \requests$ a mapping $\mapping$ is selected with probability $\prob$. Importantly, the summed probabilities may not sum to $1$, i.e. with probability $1- \sum_k \prob$ the request $\req$ is not embedded.

The rounding procedure is iterated as long as the constructed solution is not of sufficient quality or until the number of maximal rounding tries is exceeded. Concretely, we seek $(\alpha,\beta,\gamma)$-approximate solutions which achieve at least a factor of $\alpha\leq1$ times the optimal (LP) profit and exceed node and edge capacities by at most factors of $\beta\geq1$ and $\gamma \geq1$, respectively. In the following we derive parameters $\alpha$, $\beta$, and $\gamma$ for which solutions can be found \emph{with high probability}.

Note that Algorithm~\ref{alg:rand-round-profit} is indeed a polynomial-time algorithm, as the size of the novel LP Formulation~\ref{IP:novel} is polynomially bounded and can hence be solved in polynomial-time.

\subsection{Probabilistic Guarantee for the Profit} 
\label{sec:performance-guarantee-obj-admission-control}

For bounding the profit achieved by the randomized rounding scheme, we recast the profit achieved in terms of random variables. The \emph{discrete} random variable $\randVarY \in \{0,\Vprofit\}$ models the profit achieved by the rounding of request $\req \in \requests$. According to our rounding scheme, we have $\ProbVarY[\Vprofit] = \sum_{k} \prob$ and $\ProbVarY[0] = 1-\sum_{k} \prob$. We denote the overall profit by $B = \sum_{\req \in \requests} \randVarY$ with \mbox{$\mathbb{E}(B) = \sum_{\req \in \requests} \Vprofit \cdot \sum_{k} \prob $}. Denoting the profit of an optimal LP solution by $\optLP$, we have $\optLP = \Exp(B)$ due to the decomposition's completeness (cf. Theorem~\ref{thm:decomposability-novel-formulation}).

By preprocessing the requests and confirming that each request can be fully embedded, the LP will attain at least the maximal profit of any of the considered requests:

\begin{lemma}
\label{lem:expected-profit-larger-max}
$\mathbb{E}(B) = \optLP \geq \max_{\req \in \requests} \Vprofit$ holds.
\end{lemma}

We employ the following Chernoff bound over continuous variables to bound the probability of achieving a small profit.

\begin{theorem}[Chernoff Bound~\cite{dubhashi2009concentration}]
\label{thm:chernoff}
Let $X\hspace{-2pt}= \hspace{-2pt}\sum_{i = 1}^n X_i$, $X_i \in [0,1]$, be a sum of $n$ independent random variables. 
For any~\mbox{$0 < \varepsilon < 1$}, the following holds: 
{
\[\mathbb{P} \big(X \leq (1-\varepsilon)\cdot \mathbb{E}(X)\big)\leq  \exp(-\varepsilon^2\cdot \mathbb{E}(X)/2)\]
}
\end{theorem}
\begin{restatable}{theorem}{probabilityOfNotSucceedingInObjective}
\label{thm:probability-of-not-succeeding-in-objective}
Let $\optIP$ denote the profit of an optimal solution. Then $\mathbb{P}(B < 1/3\cdot \optIP) \leq \exp(-2/9) \approx 0.8007$ holds.
\end{restatable}
\begin{proof}
Let $\hat{b} = \max_{\req \in \requests} \Vprofit$ be the maximum benefit among the preprocessed requests.
We consider random variables \mbox{$Y'_{\req} = Y_{\req} / \hat{b}$}, such that $Y'_{\req} \in [0,1]$ holds.
Let \mbox{$B' = \sum_{\req \in \requests} Y'_{\req} = B / \hat{b}$}.
As $\mathbb{E}(B) = \optLP \geq \hat{b}$ holds (cf. Lemma~\ref{lem:expected-profit-larger-max}), we have $\mathbb{E}(B') \geq 1$.
Choosing $\varepsilon = 2/3$ and applying Theorem~\ref{thm:chernoff} on $B'$ we obtain \mbox{$\mathbb{P} \big( B' \leq (1/3)\cdot \mathbb{E}(B') \big) \leq \exp(-2\cdot \mathbb{E}(B')/9)$}.
Plugging in the \emph{minimal} value of $\Exp(B')$, i.e. $1$, into the equation we obtain:
$\mathbb{P} \big( B' \leq (1/3)\cdot \mathbb{E}(B')~\big) \leq \exp(-2/9)$ and by linearity 
$\mathbb{P} \big( B \leq (1/3)\cdot \mathbb{E}(B)~\big) \leq \exp(-2/9)$.

Denoting the profit of an optimal solution by $\optIP$ and observing that $\optIP \leq \optLP$ holds as the linear relaxation yields an upper bound, we have
$\optIP/3 \leq \Exp(B)/3$. Accordingly, we conclude that, $\mathbb{P} \big( B \leq (1/3) \cdot \optIP \big) \leq \exp(-2/9)$ holds.
\end{proof}

\subsection{Probabilistic Guarantee for Resource Augmentations}
\label{sec:performance-guarantee-cap-violation-admission-control}
In the following, we analyze the probability that a rounded solution exceeds substrate capacities by a certain factor. 

We first note that $\maxDemandX \leq \Scap(x,y)$ holds for all resources $(x,y) \in \SR$ and all requests $\req \in \requests$.
We model the allocations on resource $(x,y) \in  \SR$ by request $\req \in \requests$ as random variable \mbox{$A_{\req,x,y} \in [0,\maxAllocX]$}. By definition, we have $\mathbb{P}(A_{\req,x,y} = A(\mapping,x,y))= \prob $ and \mbox{$\mathbb{P}(A_{\req,x,y} = 0)=  1 - \sum_{k } \prob$}. Furthermore, we denote by $A_{x,y} = \sum_{\req \in \requests} A_{\req,x,y}$ the random variable capturing the overall allocations on resource $(x,y) \in \SR$.
As $\mathbb{E}(A_{x,y}) = \sum_{\req \in \requests} \sum_{k} \prob \cdot A(\mapping,x,y)$ holds by Theorem~\ref{thm:decomposability-novel-formulation}, we obtain $\mathbb{E}(A_{x,y}) \leq \Scap(x,y)$ for all resources $(x,y) \in  \SR$.

We employ Hoeffding's inequality to upper bound $A_{x,y}$.

\begin{theorem}[Hoeffding's inequality \cite{dubhashi2009concentration}]
Let $X\hspace{-2pt}= \hspace{-2pt}\sum_{i = 1}^n X_i$, $X_i \in [a_i,b_i]$, be a sum of $n$ independent random variables. 
The following holds for any $t \geq 0$:
{
\[
\mathbb{P}(X - \mathbb{E}(X)\geq t)\leq \exp (-2t^2 / (\sum \nolimits_i(b_i - a_i)^2))
\]
}
\end{theorem}


{
\renewcommand{\DeltaV}{\ensuremath{\Delta{(x,y)}}}
\begin{lemma}
\label{lem:approximation-single-resource} Consider a resource \mbox{$(x,y) \in \SR$} and \mbox{$0< \varepsilon \leq 1$}, such that $\maxDemandX / \Scap(x,y) \leq \varepsilon$ holds for  \mbox{$\req \in \requests$}. 
Let \mbox{$\DeltaV = \sum_{\req \in \requests: \maxDemandX > 0} (\maxAllocX / \maxDemandX)^2$}.
\begin{align}
\mathbb{P} (  A_{x,y} \geq \delta(\lambda) \cdot \Scap(x,y) )  \leq \lambda^{-4} \label{eq:foo}
\end{align}
holds for $\delta(\lambda) = 1+\varepsilon \cdot  \sqrt{2\cdot \DeltaV \cdot \log(\lambda)}$ and any $\lambda > 0$.
\begin{proof}\allowdisplaybreaks
We apply Hoeffding with~$t =( 1 - \delta(\lambda))\cdot \Scap(x,y)$:
{
\small{\noindent\begin{alignat*}{2}
& \mathbb{P} \bigg(A_{x,y} - \mathbb{E}(A_{x,y}) \geq ( 1 - \delta(\lambda))\cdot \Scap(x,y) \bigg) &&  \notag \\
& ~ \leq \exp \Big(\frac{-4 \cdot \varepsilon^2 \cdot \log (\lambda) \cdot \DeltaV \cdot  \Scap^2(x,y)}{\sum \limits_{\req \in \requests} (\maxAllocX)^2} \Big) && \\
& ~ \leq \exp \Big(\frac{-4 \cdot \varepsilon^2 \cdot \log (\lambda) \cdot \DeltaV \cdot  \Scap^2(x,y)}{\hspace{-12pt}\sum \limits_{\req \in \requests: \maxDemandX > 0} \hspace{-24pt}(\maxAllocX)^2} \Big) && \\
& ~ \leq \exp \Big(\frac{-4 \cdot \varepsilon^2 \cdot \log (\lambda) \cdot \DeltaV \cdot \Scap^2(x,y) }{\sum \limits_{\req \in \requests: \maxDemandX > 0} \hspace{-24pt} (\varepsilon \cdot \Scap(x,y) \cdot \maxAllocX / \maxDemandX)^2} \Big) && \\
& ~ \leq \exp \Big(\frac{-4 \cdot \log (\lambda) \cdot \DeltaV }{\sum \limits_{\req \in \requests: \maxDemandX > 0} \hspace{-24pt}( \maxAllocX / \maxDemandX)^2} \Big) =\lambda^{-4}
\end{alignat*}
}
}
The second inequality holds, as \mbox{$\maxAllocX > 0$} implies \mbox{$\maxDemandX > 0$}. For the third inequality,  
\mbox{$\maxAllocX \hspace{-1pt}\leq \hspace{-1pt}\varepsilon \hspace{-1pt}\cdot \hspace{-1pt}\Scap(x,y) \cdot \maxAllocX / \maxDemandX$}~is~used, which follows from the assumption $\maxDemandX \leq \varepsilon \cdot \Scap(x,y)$ and $\maxDemandX > 0$. 
In the next step, $\varepsilon^2 \cdot \Scap^2(x,y)$ is reduced from the fraction.  As the denominator equals $\DeltaV$ by definition, the final equality follows.
Lastly, we utilize that the expected allocation $\mathbb{E}(A_{x,y})$ is upper bounded by the resource's capacity $\Scap(x,y)$ to obtain Equation~\ref{eq:foo}.
\end{proof}
\end{lemma}
}

Given Lemma~\ref{lem:approximation-single-resource}, we obtain the following corollary.
\begin{corollary} 
\label{cor:resource-augmentation-prob}
Let $0 < \varepsilon \leq 1$ be chosen minimally, such that $\maxDemandX / \Scap(x,y) \leq \varepsilon$ holds for all resources $(x,y) \in \SR$ and all requests $\req \in \requests$. Let $\Delta(X)=\max_{(x,y) \in X} \Delta(x,y)$,
\begin{alignat*}{8}
\beta   & = && (1+\varepsilon \cdot  \sqrt{2\cdot \Delta(\SRV) \cdot \log(|\SV|\cdot |\types|)})  \textit{~,~and}\\
\gamma & =  && (1+\varepsilon \cdot  \sqrt{2\cdot \Delta(\SE) \cdot \log(|\SE|)})\,. \\[-18pt]
\end{alignat*}
The following holds for all node resources $(\type,u) \in \SRV$ and edge resources $(u,v) \in \SE$, respectively:
\begin{alignat}{8}
\mathbb{P} (  A_{\type,u} \geq  \beta \cdot \Scap(\type,u) )  & \leq && {(|\SV|\cdot |\types|)}^{-4} \label{eq:req-node}\\
\mathbb{P} (  A_{u,v} \geq  \gamma \cdot \Scap(u,v) ) & \leq && {|\SE|}^{-4} \label{eq:req-edge}
\end{alignat}
\end{corollary}
\begin{proof}
First, note that $\varepsilon$ is chosen over all resources and requests and that $\Delta(\SRV) \geq \Delta(\type,u)$ and $\Delta(\SE) \geq \Delta(u,v)$ hold for $(\type,u) \in \SRV$ and  $(u,v) \in \SE$, respectively. 
 Equations~\ref{eq:req-node} and \ref{eq:req-edge} are then obtained from Lemma~\ref{lem:approximation-single-resource} by setting \mbox{$\lambda = |\SV| \cdot |\types|$} for nodes and $\lambda = |\SE|$ for edges.
\end{proof}

\subsection{Approximation Result}
\label{sec:main-results-admission-control}

Given the probabilistic bounds established above, the main approximation result is obtained via a union bound.

\begin{restatable}{theorem}{mainResultForAdmissionControl}
\label{thm:result-for-admission-control}
Assume $|\SV| \geq 3$. Let $\beta$ and $\gamma$ be defined as in Corollary~\ref{cor:resource-augmentation-prob}.
Algorithm~\ref{alg:rand-round-profit} returns $(\alpha,\beta,\gamma)$-approximate solutions for the VNEP (restricted on cactus request graphs) of at least an~$\alpha = 1/3$ fraction of the optimal profit, and allocations on nodes and edges within factors of $\beta$ 
 and $\gamma$ of the original capacities, respectively, with high probability.
\end{restatable}
\begin{proof}
We employ the following union bound argument. Employing Corollary~\ref{cor:resource-augmentation-prob} and as there are at most $|\SV| \cdot |\types|$ node resources and at most $|\SV|^2$ edges, the \emph{joint} probability that any resource exceeds their respective capacity by factors of $\beta$ or $\gamma$ is upper bounded by $(|\SV|\cdot|\types|)^3 + |\SV|^2 \leq 1/27 + 1/9$ for $|\SV| \geq 3$. 
By Theorem~\ref{thm:probability-of-not-succeeding-in-objective} the probability of \emph{not} finding a solution achieving an $\alpha=1/3$ fraction of the optimal objective is upper bounded by $\exp(-2/9)$. Hence, the probability to not find a $(\alpha,\beta,\gamma)$-approximate solution within a single round is upper bounded by~$\exp(-2/9) + 1/9 + 1/27 \leq 19/20$. Hence, as the probability to return a suitable solution within a single round is at least $1/20$, the probability that the algorithm returns an approximate solution within $N \in \mathbb{N}$ rounding tries is lower bounded by $1-(19/20)^N$. Thus,  Algorithm~\ref{alg:rand-round-profit} yields approximate solutions for the VNEP \emph{with high probability}.
\end{proof}

\subsection{Discussion \& Proposed Heuristics}
\label{sec:main-section-discussion}
Theorem~\ref{thm:result-for-admission-control} yields the first  approximation algorithm for the profit variant of the VNEP. However, the direct application of Algorithm~\ref{alg:rand-round-profit} to compute $(\alpha,\beta,\gamma)$-approximate solutions is made difficult by the cumbersome definition of the terms $\Delta(\SRV)$ and $\Delta(\SE)$. Specifically, computing $\beta$ and $\gamma$ exactly requires enumerating all valid mappings, which is not feasible. Hence, to directly apply Algorithm~\ref{alg:rand-round-profit}, the respective values have to be estimated. 
The following upper bounds can be easily established: 
\begin{alignat}{7}
\Delta(\SRV) & = && |\requests| \cdot \max_{\req \in \requests} |\VV| \\
\Delta(\SE) & \leq && |\requests| \cdot \max_{\req \in \requests} |\VE|\,.
\end{alignat}

Both inequalities follow from the observation that $\maxAllocX / \maxDemandX$ is upper bounded by the number of virtual nodes and edges, respectively, as $\maxAllocV \leq \maxDemandV \cdot |\VV|$ and $\maxAllocE \leq \maxDemandE \cdot |\VE|$ holds for node resources $(\type,u) \in \SRV$ and edge resource $(u,v) \in \SE$, respectively.
Now, plugging these upper bounds into the definition of $\beta$ and $\gamma$ yields large resource augmentations of  \mbox{$\beta \in \mathcal{O}(\varepsilon \cdot \sqrt{|\requests| \cdot  \max_{\req \in \requests} |\VV| \cdot \log(|\SV| \cdot |\types|)})$} and \mbox{$\gamma \in \mathcal{O}(\varepsilon \cdot \sqrt{|\requests| \cdot  \max_{\req \in \requests} |\VE| \cdot \log(|\SE|)})$}, respectively.

\begin{figure}[t!]
\centering
\includegraphics[height=0.155\textheight]{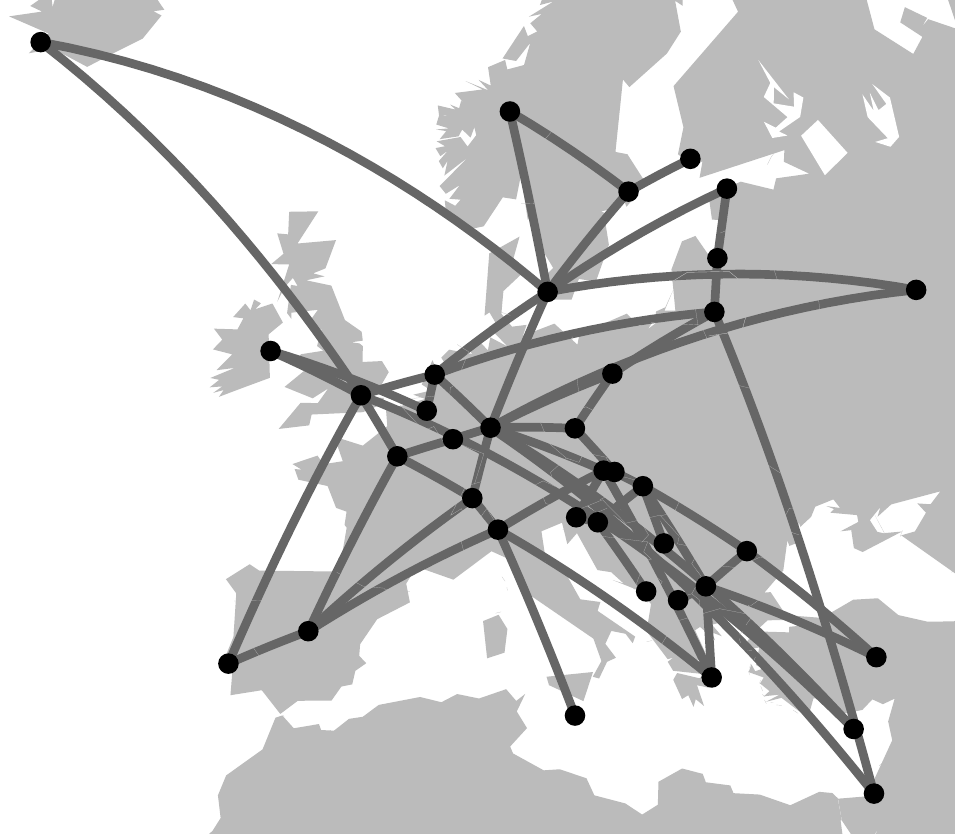}
\caption{Visualization of the G\'{E}ANT network chosen as substrate network. The G\'{E}ANT network connects several research, education and innovation institutes across Europe. }
\label{fig:geant}
\end{figure}

To overcome estimating $\beta$ and $\gamma$ by the above (or similar) bounds, we propose the following to apply randomized rounding in practice.
Concretely, we consider two types of adaptations to our Algorithm~\ref{alg:rand-round-profit}. The first adaptation is to discard the consideration of respective approximation factors $\alpha,\beta,\gamma$ and simply returning \emph{the best} solution found within a fixed number of rounding iterations. We refer to this type of rounding heuristic as \emph{vanilla rounding}, as the rounding procedure itself is not changed. The other type of adaptation is the avoidance of resource augmentations by considering \emph{heuristical rounding}.

Concretely, we consider two \emph{vanilla rounding} heuristics, namely $\RRMIN$ and $\RRMAX$, as well as one \emph{heuristical rounding} variant $\RRHEUR$, as described below:

\begin{description}
\item[$\RRMIN$] A fixed number of solutions  is rounded at random as in Line~\ref{alg:randround:rounding-main} of Algorithm~\ref{alg:rand-round-profit}. The solution minimizing resource augmentations (and among those the one of the highest profit) is selected. 
\item[$\RRMAX$]  A fixed number of solutions  is rounded at random as in Line~\ref{alg:randround:rounding-main} of Algorithm~\ref{alg:rand-round-profit}. The solution maximizing the profit (and among those the one of the least resource augmentations) is selected. 
\item[$\RRHEUR$]  Line~\ref{alg:randround:rounding-main} of Algorithm~\ref{alg:rand-round-profit} is adapted in such a way that selected mappings are only accepted, when its incorporation into the solution does not exceed resource capacities. In other words, if a mapping $\mapping$ is selected for request $\req \in \requests$ whose addition would exceed any resource capacity, then this mapping is simply discarded and request $\req$ is not embedded. To increase the diversity of found solutions, the order in which requests are processed is permuted before each rounding iteration. A fixed number of solutions is rounded and the one maximizing the profit is returned.
\end{description}

\section{Explorative Computational Study}
\label{sec:evaluation}

We now complement our formal approximation result in
the standard multi-criteria model with resource augmentation with an extensive computational study. Specifically, we study the performance of vanilla rounding and heuristical rounding (i.e., without resource augmentations) as introduced above.

As we are not aware of any systematic
evaluation of the profit maximization in the offline settings,
we present a synthetic but extensive computational study. 
Specifically, we have generated 1,500 offline VNEP instances with varying request numbers and varying demand-to-capacity ratios. For all instances, baseline solutions were computed by solving the (Mixed-)Integer Programming Formulation~\ref{alg:VNEP-IP-old}.

We have implemented all presented algorithms in Python 2.7 employing Gurobi 7.5.1 to solve the Mixed-Integer Programs and Linear Programs. Our source code is freely available at~\cite{github-evaluation}. All experiments were executed on a server equipped with Intel Xeon E5-4627v3 CPUs running at 2.6 GHz and reported runtimes are wall-clock times.

\subsection{Instance Generation}
\newcommand{\NRF}{\ensuremath{\mathrm{NRF}}}
\newcommand{\ERF}{\ensuremath{\mathrm{ERF}}}

\paragraph{Substrate Graph}
We use the G\'EANT topology\footnote{Obtained from \url{http://www.topology-zoo.org/} (version March 2012)\,.} as substrate network. It consists of 40 nodes and 122 edges (see Figure~\ref{fig:geant}). We consider a single node type and set node and edge capacities uniformly to $100$.

\begin{figure}[t!]
\centering
\includegraphics[height=0.155\textheight]{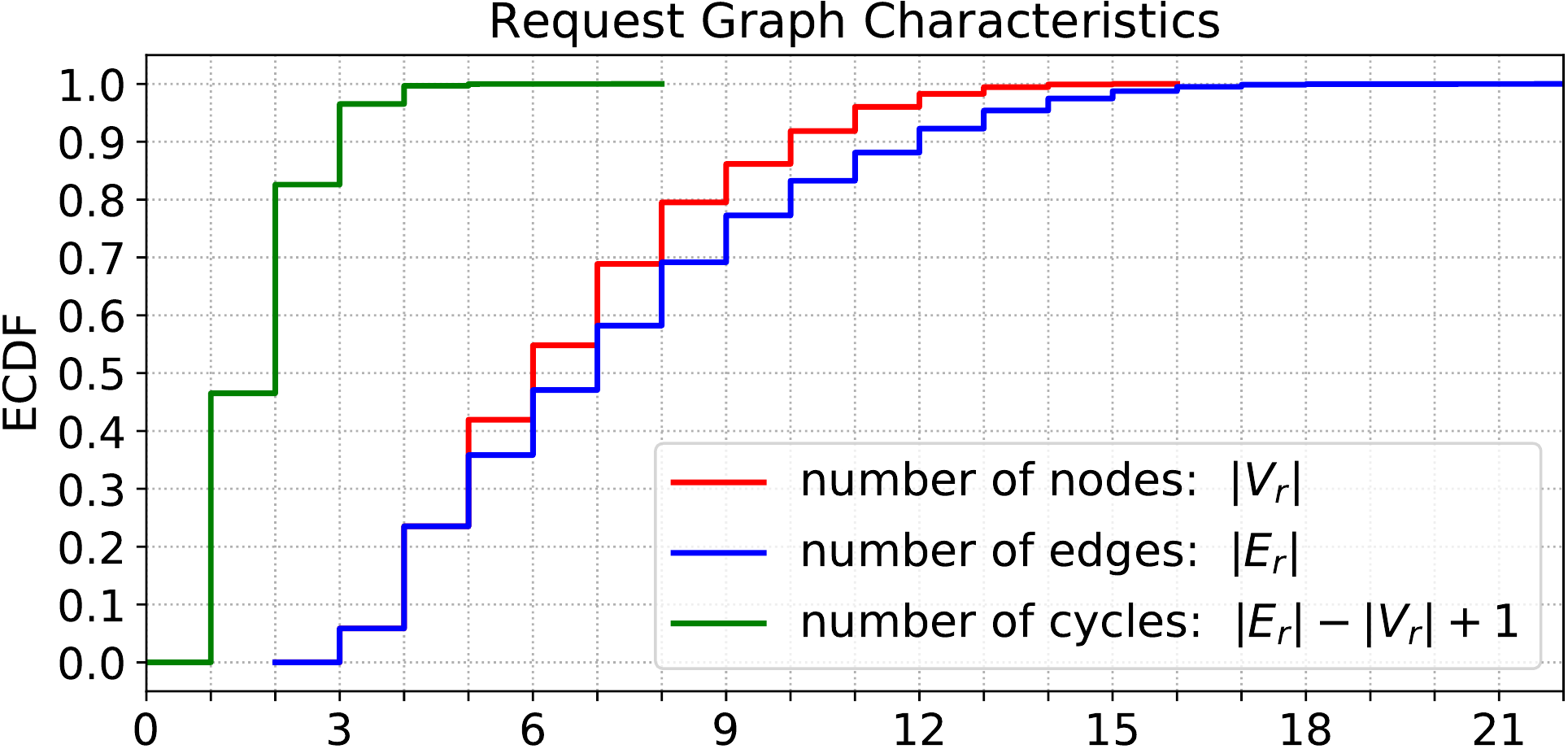}
\caption{Characteristics of the generated cactus requests graphs, namely the number of nodes, edges and number of cycles. The depicted empirical cumulative distribution function (ECDF) is based upon 100k sampled cactus requests graphs.}
\label{fig:request-characteristics}
\end{figure}

\paragraph{Request Topology Generation}
Cactus graph requests are generated by (i) sampling a random binary tree of maximum depth 3, (ii) adding additional edges randomly as long as they do not refute the cactus property \emph{as long as such edges exist}, and (iii) orienting edges arbitrarily.

\begin{figure*}[t!]
\begin{minipage}[t]{0.49\textwidth}
\includegraphics[height=.142\textheight]{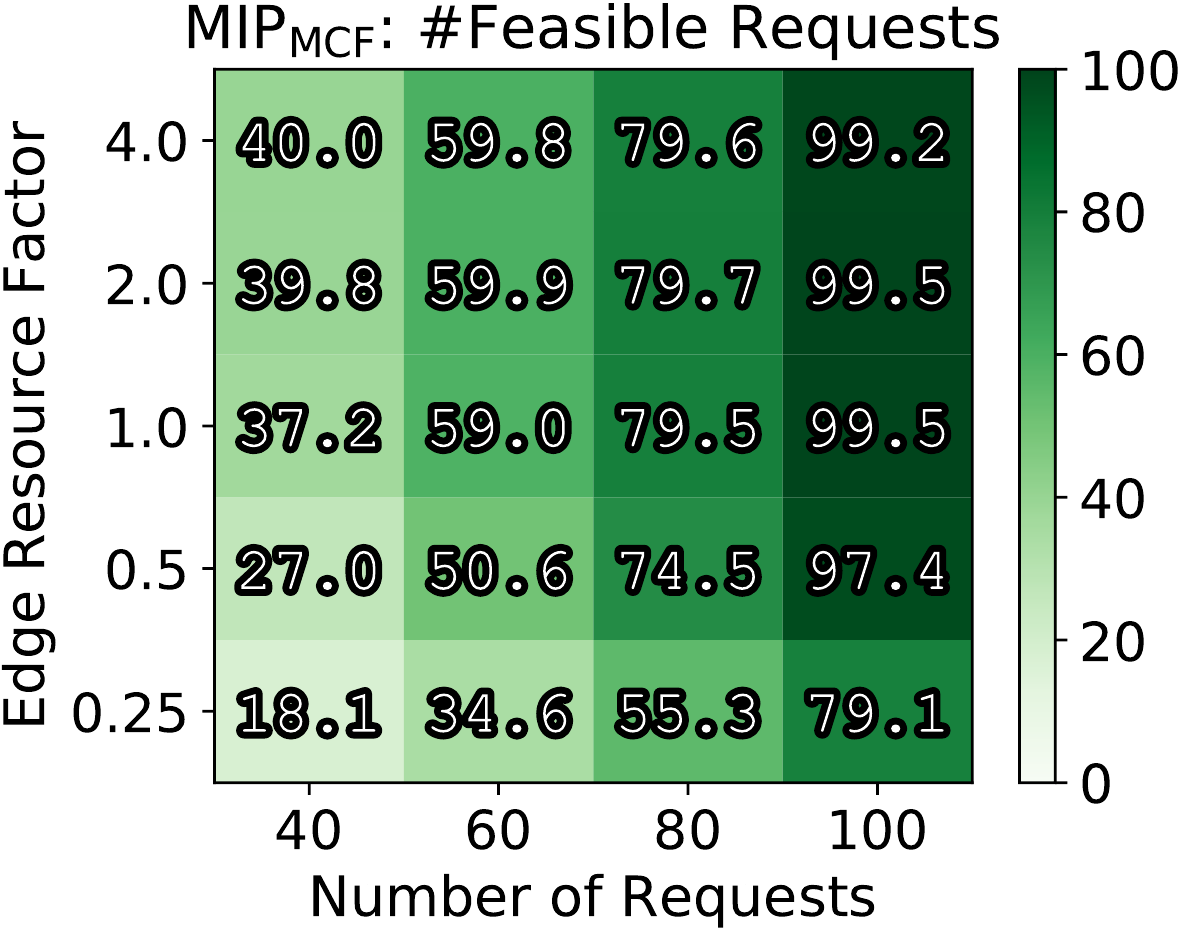}
\hfill
\includegraphics[height=.142\textheight]{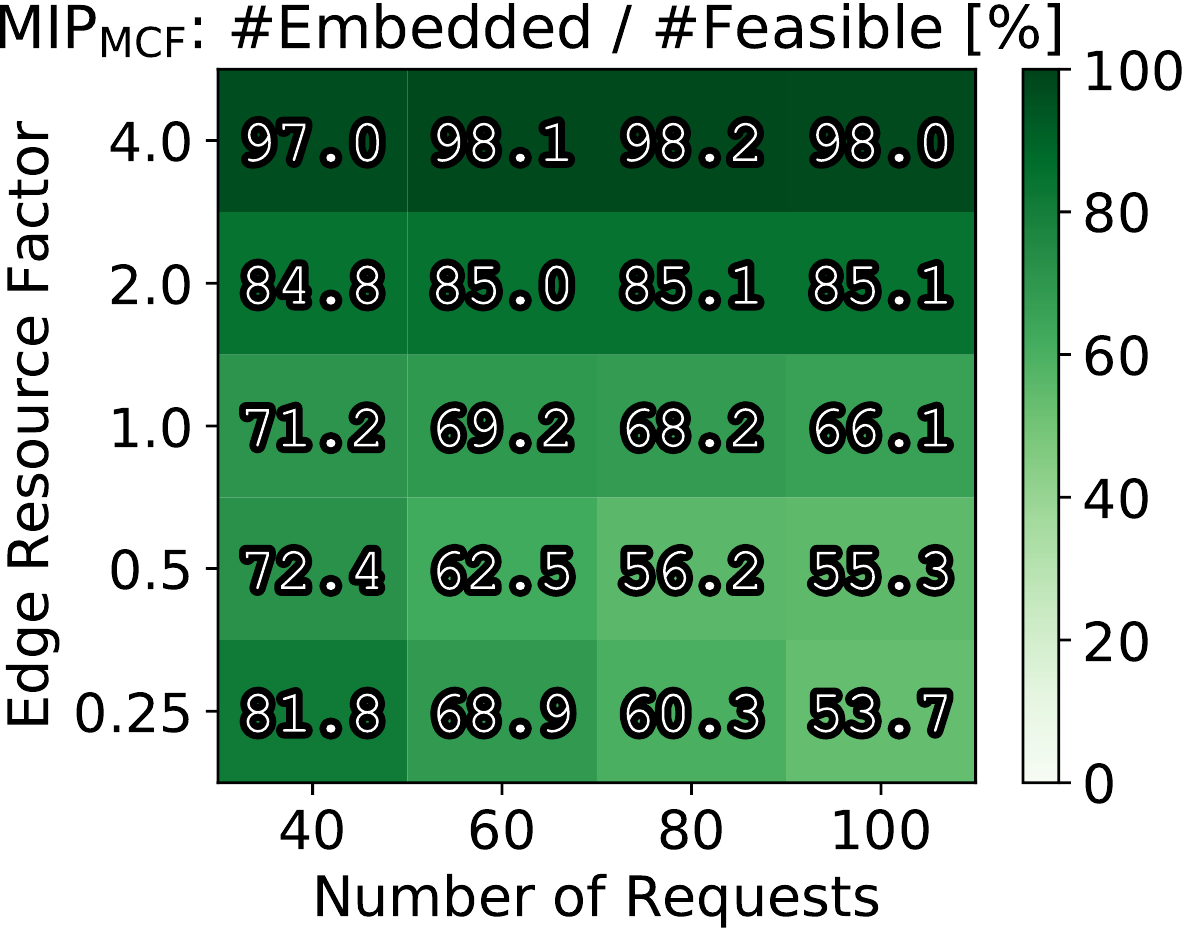}
\caption{Overview of the feasibility of generated requests and the baseline's acceptance ratio. Each cell averages the result of 75 instances. \newline
Left: The feasibility of requests is obtained from (cost-optimally) embedding the requests to compute the profit a priori. Note absolute numbers are depicted. \newline
Right: The acceptance ratio of the baseline $\MIPMCF$ with respect to the requests that were found to be feasible (infeasible requests are not considered).}
\label{fig:baseline-acceptance}
\end{minipage}
\hfill
\begin{minipage}[t]{0.49\textwidth}
\includegraphics[height=.142\textheight]{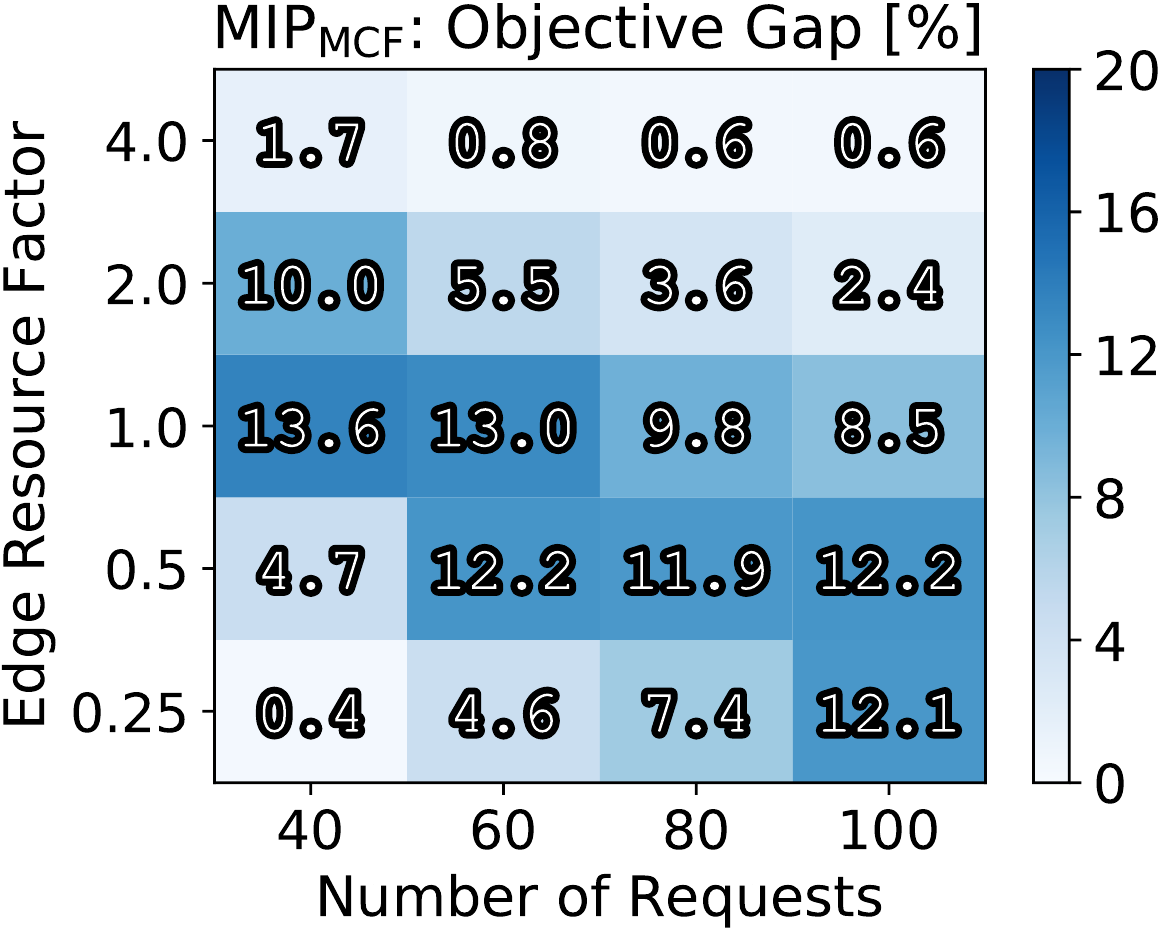}
\hfill
\includegraphics[height=.142\textheight]{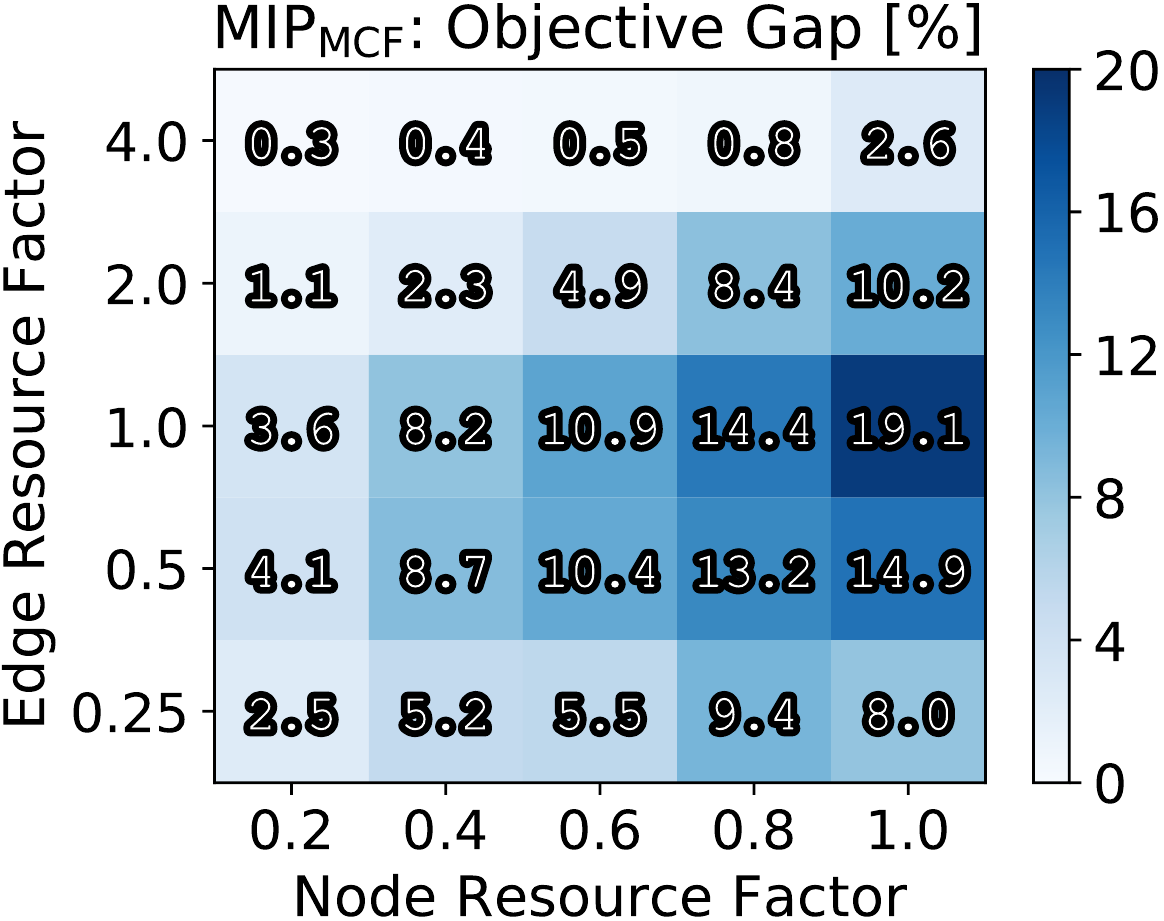}
\caption{Overview of the objective gap achieved by the baseline algorithm $\MIPMCF$ after upto 3 hours of computation time. Note the different x-axes. The left plot averages the results of 75 and the right one the results of 60 instances per cell. The number of requests, i.e. the problem size, has a less distinct impact on the objective gap than the resource factors alone.}
\label{fig:baseline-objective-gap}
\end{minipage}
\end{figure*}

Concretely, the sampling process of binary trees works as follows: starting with an initial root node, the number of children is drawn using the discrete distribution \mbox{$\mathbb{P}(\mathrm{\#children}=0) = 0.15$}, \mbox{$\mathbb{P}(\mathrm{\#children}=1) = 0.5$}, and \mbox{$\mathbb{P}(\mathrm{\#children}=2) = 0.35$}. For each (newly) generated node (of depth less than 3) further children are generated according to the same distribution. We discard graphs having less than 3 nodes.
According to the above generation procedure, the  expected number of nodes and edges is $6.54$ and $7.28$, respectively. On average, $61\%$ of the edges lie on a cycle. Figure~\ref{fig:request-characteristics} offers a more in-depth view on the request characteristics.

\paragraph{Mapping Restrictions}
To force the virtual networks to span across the whole substrate network, we restrict the mapping of virtual nodes to one quarter of the substrate nodes. Hence, the mapping of virtual nodes is restricted to ten substrate nodes. The mapping of virtual edges is not restricted.

\begin{figure}[t!]
\includegraphics[width=.24\textwidth]{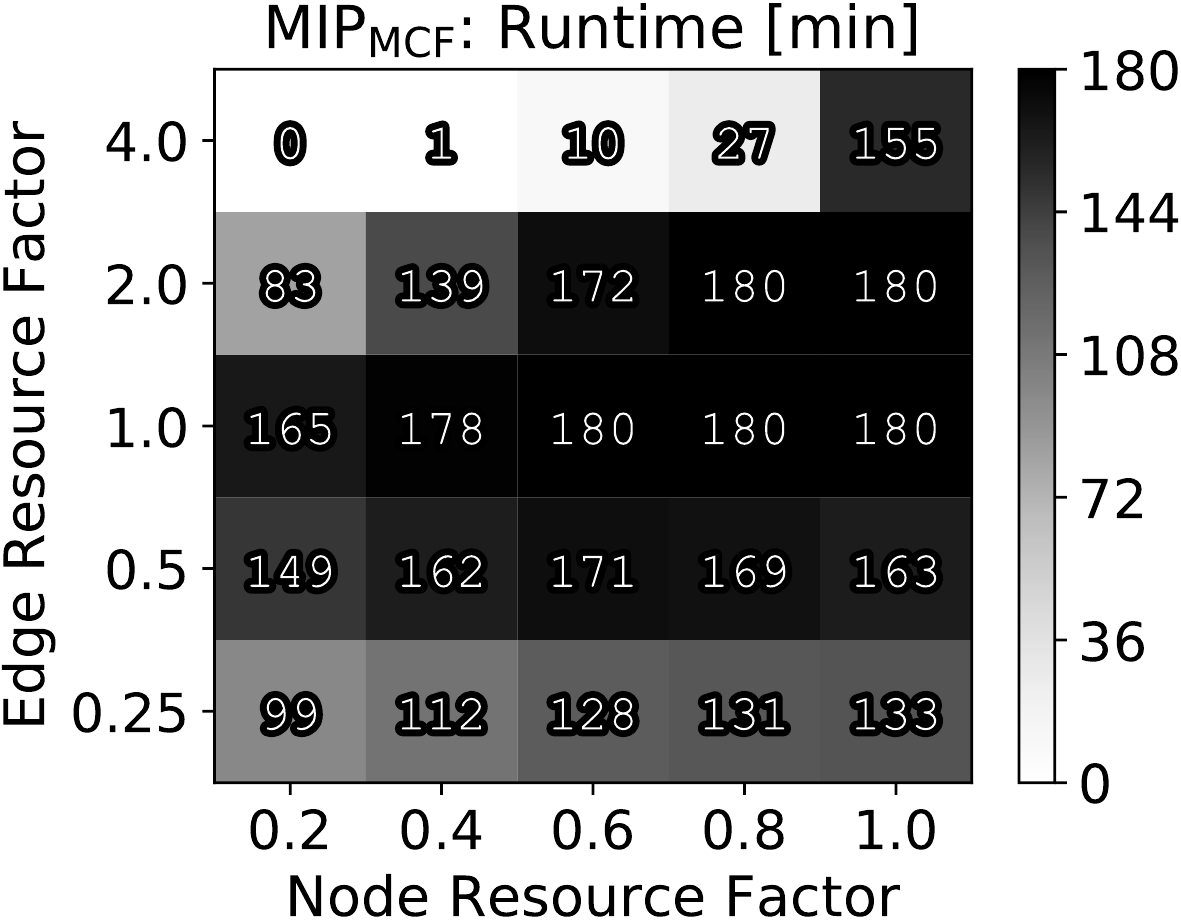}
\hfill
\includegraphics[width=.24\textwidth]{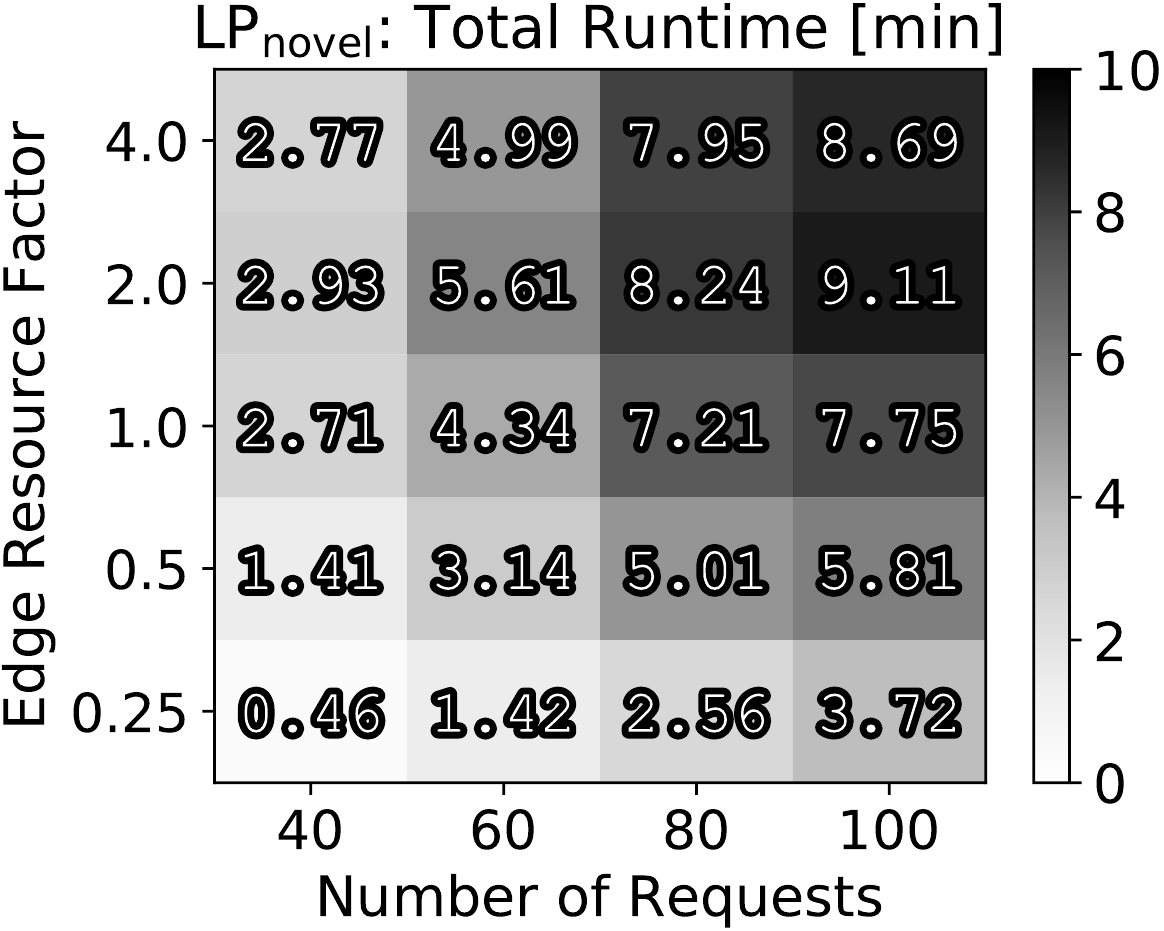}
\caption{Overview of the runtimes of the baseline algorithm $\MIPMCF$ and our novel LP formulation $\LPNOVEL$. Note the different x-axes. On the left each cell averages 60 results while on the right 75 results are averaged per cell.\newline
Left: The runtime of the MCF Formulation~\ref{alg:VNEP-IP-old} being solved by Gurobi 7.5.1. Copmutations are terminated after 180 minutes.\newline
Right: The runtime of the novel LP Formulation~\ref{IP:novel} -- solved using the Barrier algorithm of Gurobi 7.5.1. -- including the constructing time of the LP. }
\label{fig:runtimes}
\end{figure}

\paragraph{Demand Generation}

We control the demand-to-capacity ratio of node and edge resource using a node resource factor $\NRF$ and an edge resource factor $\ERF$. The request's demands are drawn from an exponential distribution and afterwards normalized, such that the following holds:

{
\small
{\begin{alignat}{6}
\sum \nolimits_{\req \in \requests} \sum \nolimits_{ i \in \VV} \Vcap(i) & = && \NRF \cdot \sum \nolimits_{u \in \SV} \Scap(u) \label{eq:node-resource-factor}\\
\ERF \cdot \sum \nolimits_{\req \in \requests} \sum \nolimits_{(i,j) \in \VE} \Vcap(i,j) & = && \sum \nolimits_{(u,v) \in \SE} \Scap(u,v) \label{eq:edge-resource-factor}
\end{alignat}}
}

The resource factors can be best understood under the assumption that all requests are embedded. Under this assumption, a resource factor $\NRF=0.6$ implies that the node load -- averaged over all substrate nodes -- equals exactly~$60\%$. As virtual edges can be mapped on arbitrarily long paths (even of length $0$), the edge resource factor should be understood as follows: the $\ERF$ equals `the number of substrate edges that each virtual edge may use'. In particular, a factor $\ERF=0.5$ implies that if \emph{each} virtual edge spans \emph{exactly} $0.5$ substrate edges (and all requests are embedded), then the (averaged) edge resource utilization equals exactly~$100\%$. Hence, while increasing the $\NRF$ renders node resources more scarce, increasing the $\ERF$ reduces edge resource scarcity.

\paragraph{Profit Computation}
To correlate the profit of a request with its size, its resource demands, and its mapping restrictions, we compute for each request its minimal embedding costs as follows.
The cost $c(u,v)$ of using an edge $(u,v) \in \SV$ equals the geographical distance of its endpoints. The cost of nodes is set uniformly to $c(\cdot,u)= \sum_{(u,v) \in \SE} c(u,v) / |\SV|$ for all $u \in \SV$. Hence, the total node cost equals the total edge cost. Defining the cost of a mapping $\map$ to be $\sum_{(x,y) \in \SR} A(\map,x,y) \cdot c(x,y)$, we compute the minimum cost embedding for each request $\req \in \requests$ using an adaption of Mixed-Integer Program~\ref{alg:VNEP-IP-old} and set $\Vprofit$ accordingly.

\paragraph{Parameter Space}

We consider the following parameters
$|\requests| \in \{40,60,80,100\}$, $ \NRF \in \{0.2,0.4,0.6,0.8,1.0\}$, $\ERF \in \{0.25,\,0.5,1.0,2.0, 4.0\}$ and generate $15$ instances per parameter combination,  yielding $1,500 $ instances overall.

\subsection{Computational Results}
We first present our baseline results computed by solving the MCF Integer Programming Formulation~\ref{alg:VNEP-IP-old} and then study the performance of vanilla rounding and heuristical rounding.

\paragraph{Baseline $\MIPMCF$}
To obtain a near-optimal baseline solution for each of the 1,500 instances, we employ Gurobi 7.5.1 to solve the Mixed-Integer Programming Formulation~\ref{alg:VNEP-IP-old}  (using a single thread). We terminate the computation after 3 hours or when the \emph{objective gap} falls below $1\%$, i.e. when the constructed solution is \emph{provably} less than $1\%$ off the optimum. On average the runtime per instance is $129.8$ minutes~\mbox{(cf. Figure~\ref{fig:runtimes}, left)}.

\begin{figure*}[tb]
\includegraphics[height=.142\textheight]{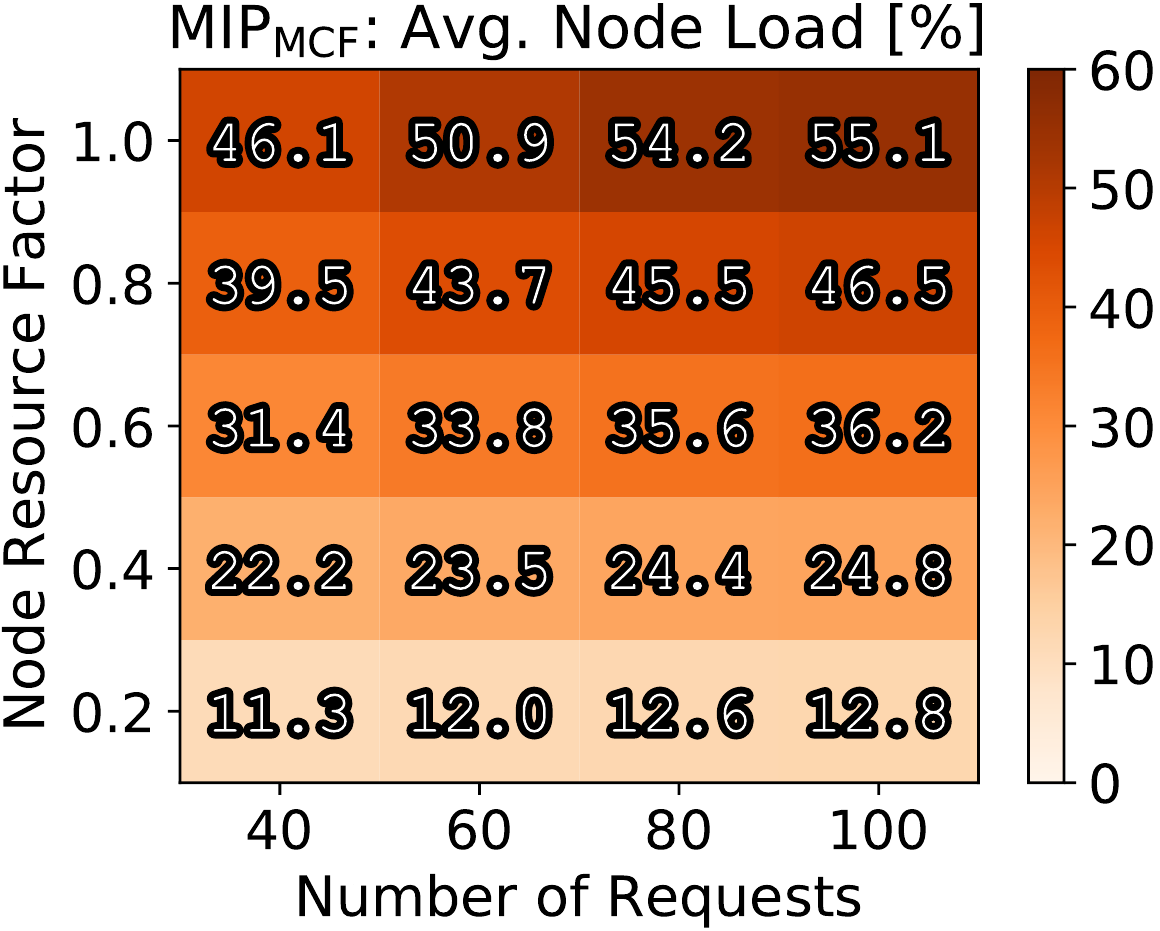}
\hfill
\includegraphics[height=.142\textheight]{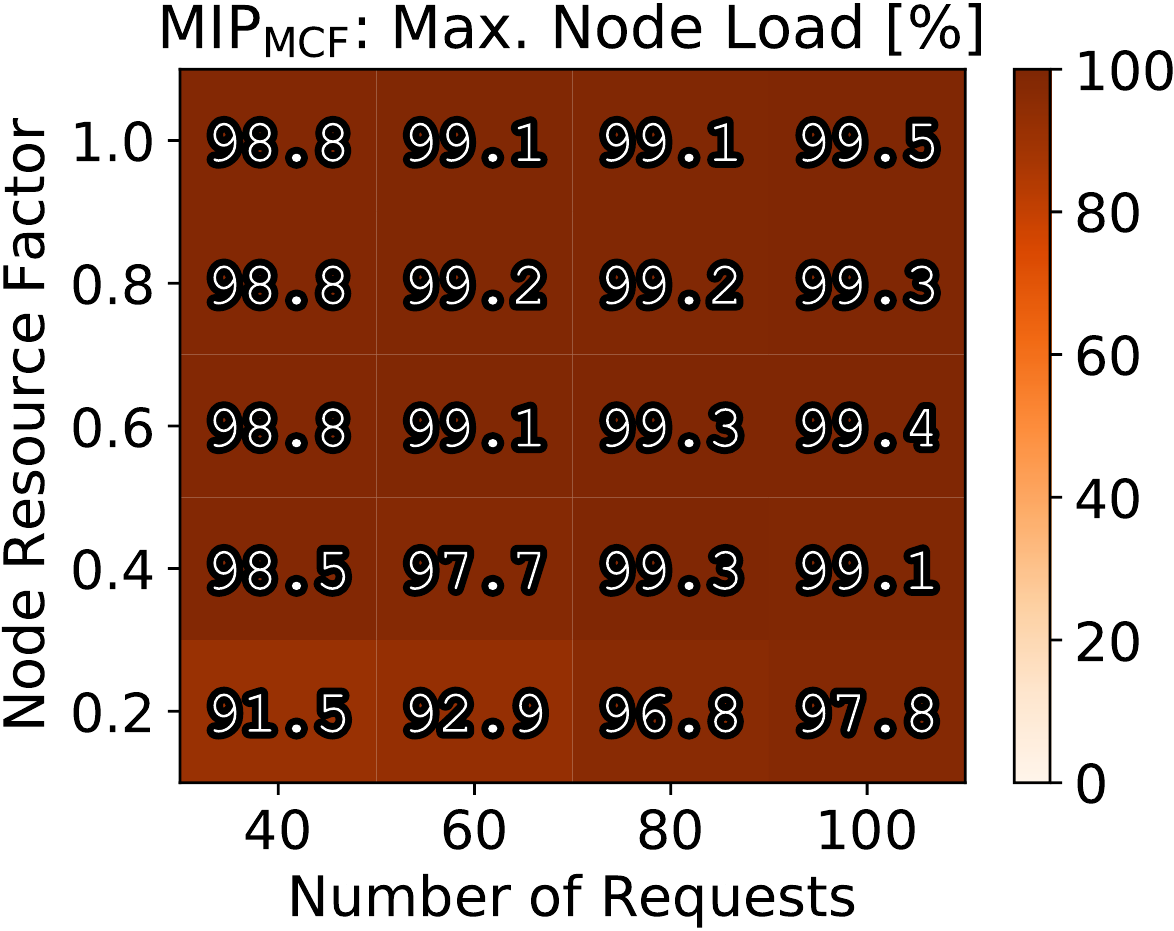}
\hfill
\includegraphics[height=.142\textheight]{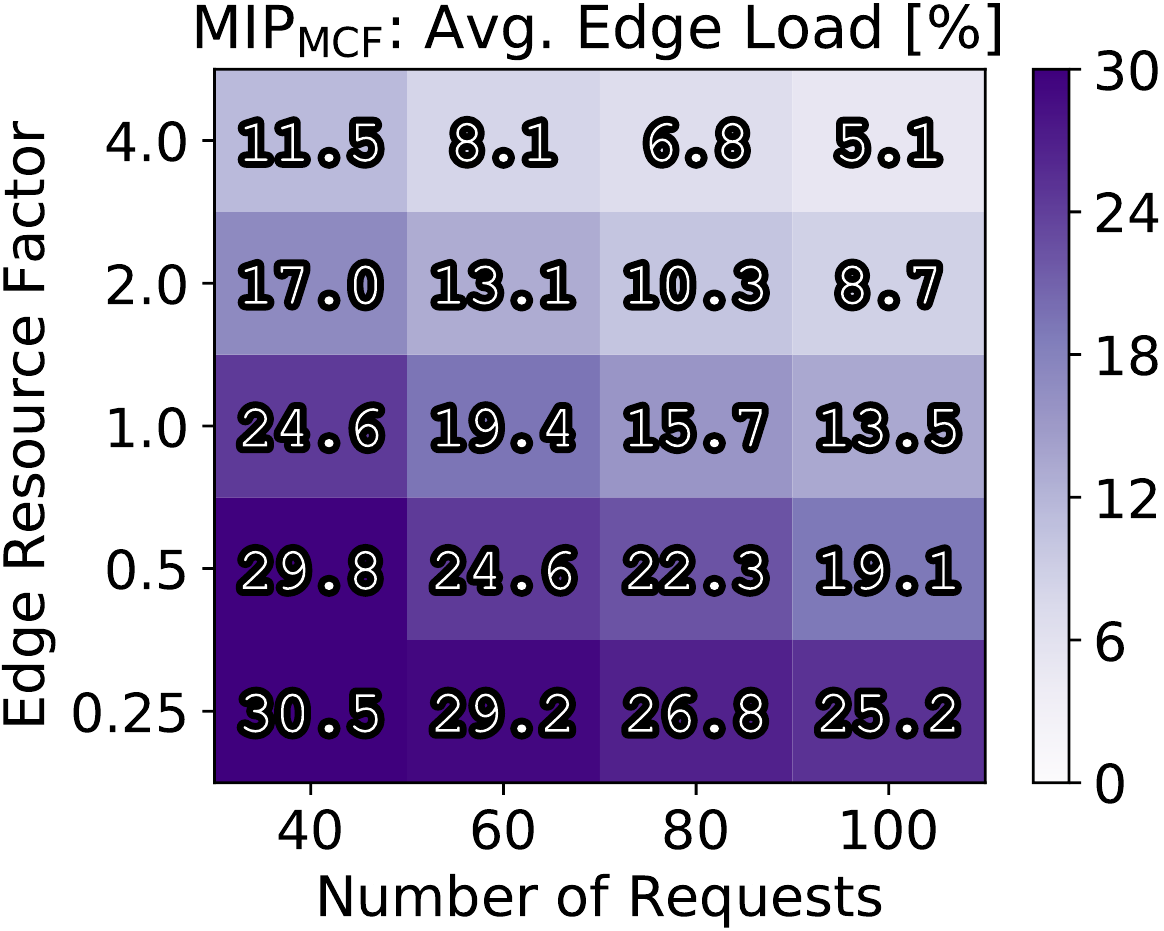}
\hfill
\includegraphics[height=.142\textheight]{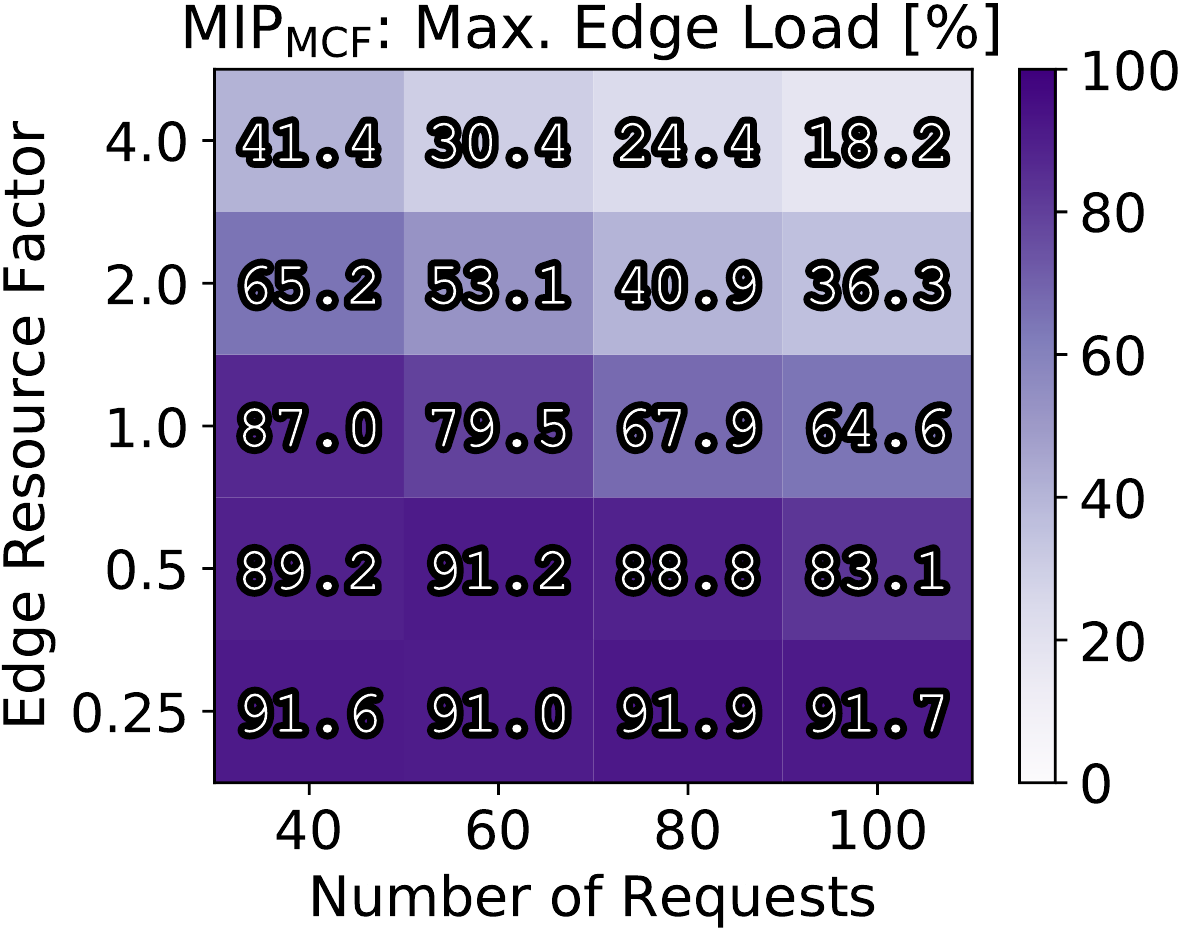}
\caption{Overview of the resource loads of the solutions computed by the baseline algorithm $\MIPMCF$. Each cell averages the results of 75 solutions.
Depicted are the averaged node/edge loads and the averaged maximum node/edge loads as a function of the node/edge resource factor and the number of requests. The respective resource factors have a distinct impact on the loads. Node resources are in general more scarce than edge resources. }
\label{fig:baseline-validation-loads}
\end{figure*}

\begin{figure*}[tb]
\begin{minipage}[t]{0.49\textwidth}
\includegraphics[height=.142\textheight]{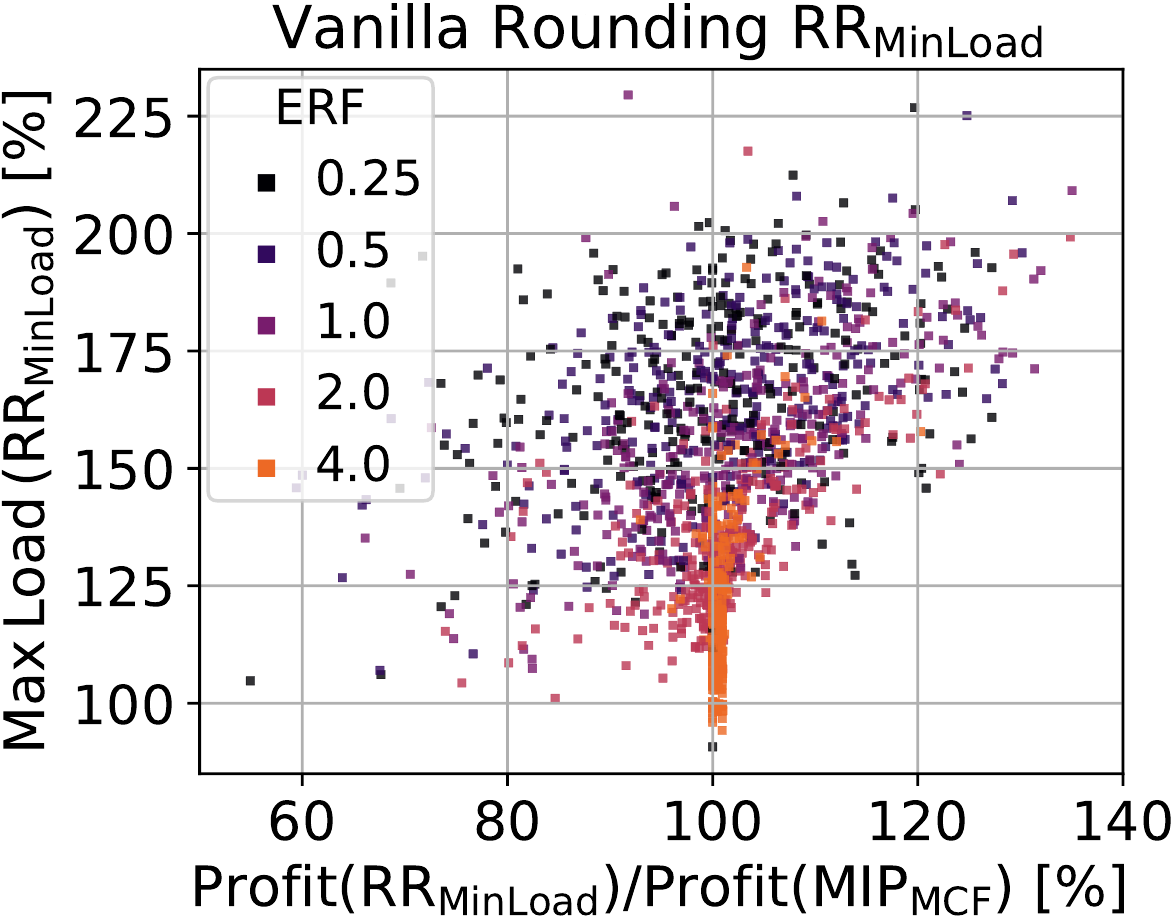}
\hfill
\includegraphics[height=.142\textheight]{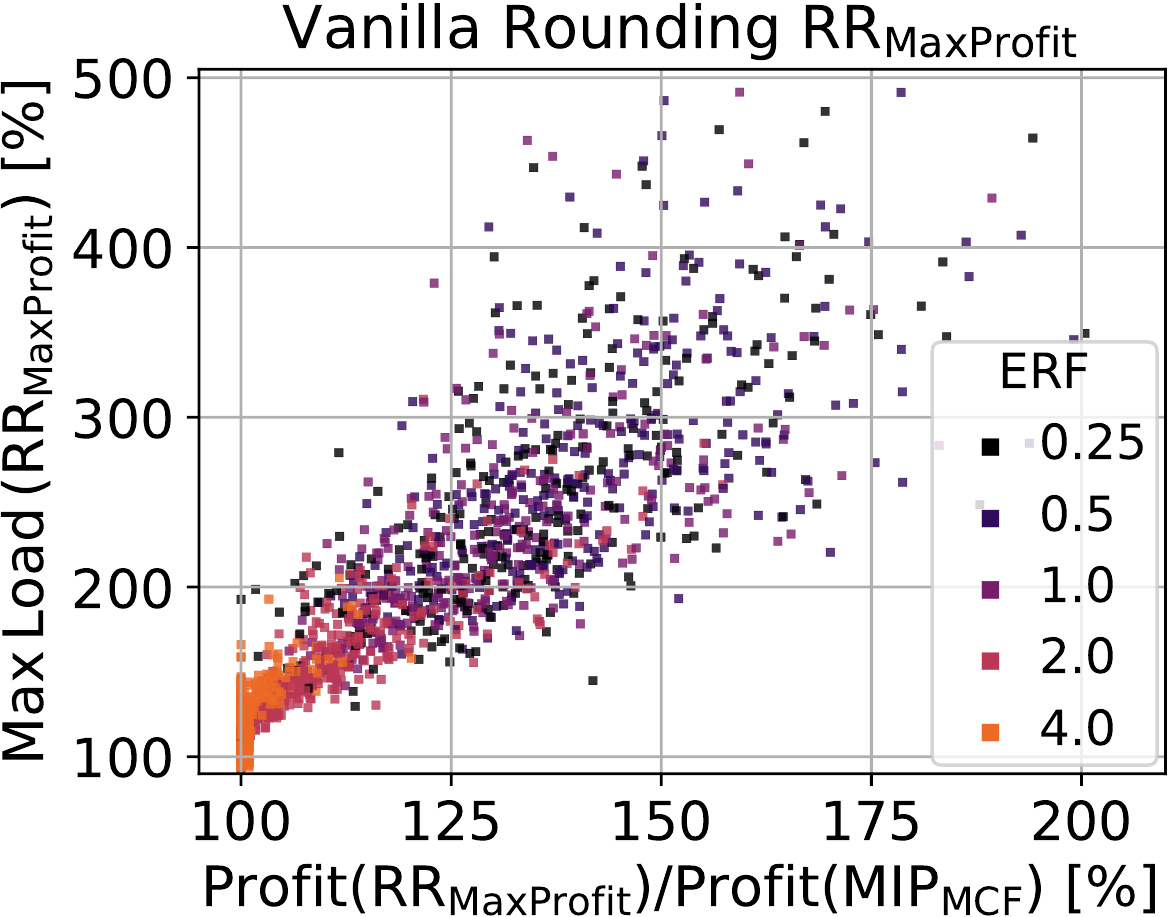}
\caption{Depicted are the solutions obtained via the two \emph{vanilla} rounding schemes $\RRMIN$ (left) and $\RRMAX$ (right). Each point corresponds to a single instance and is colored according to the instance's edge resource factor.
The left and right plots shows results for 1,493 and 1,499 of the 1,500 requests, respectively; the other results lie outside the depicted area.}
\label{fig:randround-vanilla}
\end{minipage}
\hfill
\begin{minipage}[t]{0.49\textwidth}
\includegraphics[height=.142\textheight]{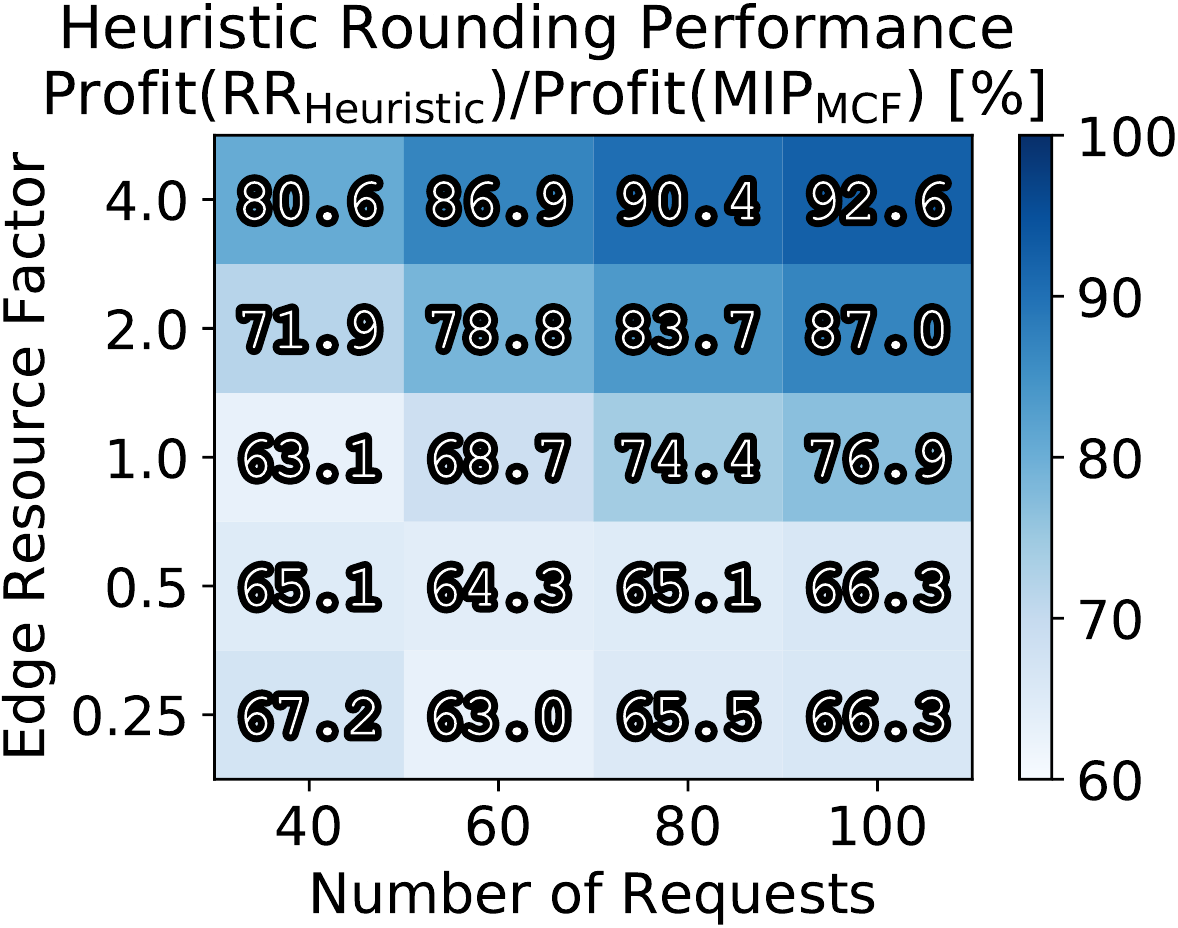}
\hfill
\includegraphics[height=.142\textheight]{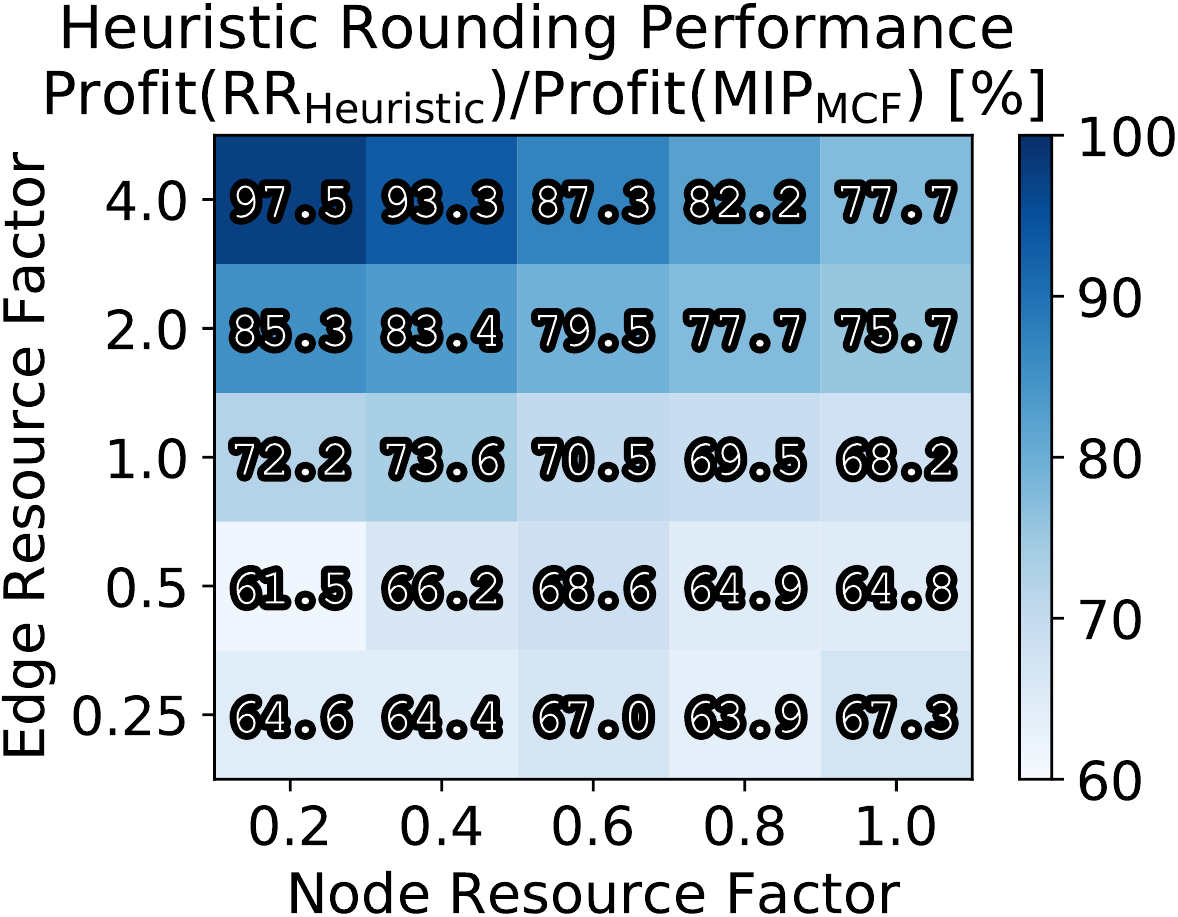}
\caption{Overview of the averaged relative performance achieved by the heuristic rounding algorithm as a function of the edge resource factor and the number of requests (left) and the node resource factor (right). Each cell averages the result of 75 (left) and 60 (right) results. The performance of the heuristic depends both on the number of requests and the resource availability.}
\label{fig:randround-results}
\end{minipage}
\vspace{-6pt}
\end{figure*}

Figure~\ref{fig:baseline-acceptance} gives an initial overview of the number of requests for which feasible embeddings exist and the acceptance ratio of the best solution as a function of the number of requests and the edge resource factor.
The number of feasible requests is determined during the a priori profit computation and may (on average) lie below 50\% when edge resources are very scarce ($\ERF = 0.25$) but otherwise consistently lies above 75\%. Similarly, the acceptance ratio of the baseline solution highly depends on the edeg resource factor, ranging from close to 53\% to roughly 98\% (on average).

Figure~\ref{fig:baseline-objective-gap} depicts quality guarantees for the baseline solutions obtained during the solution process of IP Formulation~\ref{alg:VNEP-IP-old}. Concretely, the formulation is solved using Gurobi's branch-and-bound implementation, consistently yielding upper bounds on the attainable profit (by considering the LP relaxation). Accordingly, the objective gap depicted in Figure~\ref{fig:baseline-objective-gap} gives guarantees on how far (at most) the found solutions are off optimality (on average). While increasing the number of requests does not increase the objective gap per se, both the node and edge resource factors have a distinct impact. Particularly, for the \emph{maximal} node resource factor of $1.0$ and a \emph{medium} edge resource factor of $1.0$, the averaged objective gap lies slightly above 19\%. The mean objective value is less than 7\% and the maximum observed gap across all 1,500 instances was roughly 59\%. 

Figure~\ref{fig:baseline-validation-loads} validates the impact the node and edge resource factors have on the respective resource loads. Depicted are the averaged and maximal node and edge loads as a function of the respective resource factors and the number of requests. As can be clearly seen, increasing the node resource factor increases the averaged node resource loads, while the maximum node resource load is always close to 100\%. For the edge resource loads, the picture is a different one. In particular, the averaged edge load lies significantly below the averaged node load. However, the impact of the edge resource factor is again apparent. Interestingly, for edge resource factors of 2.0 and 4.0, i.e. when edge resources are abundant, the maximal edge load only lies around 40\%, indicating that in these cases mostly the node resources are constraining the solution space.

The above can also be observed by the maximum load ECDFs depicted in Figure~\ref{fig:ecdfs-load-and-performance} (left). In particular, the maximum edge load of the baseline MIP lies consistently below the maximum node load: While 40\% of the solutions exhibit a maximum edge load of 60\% or less, only roughly 10\% of the solutions exhibit a maximum node load of 90\% or less.

\begin{figure*}[tb]
\begin{minipage}[b]{0.49\textwidth}
\includegraphics[height=.142\textheight]{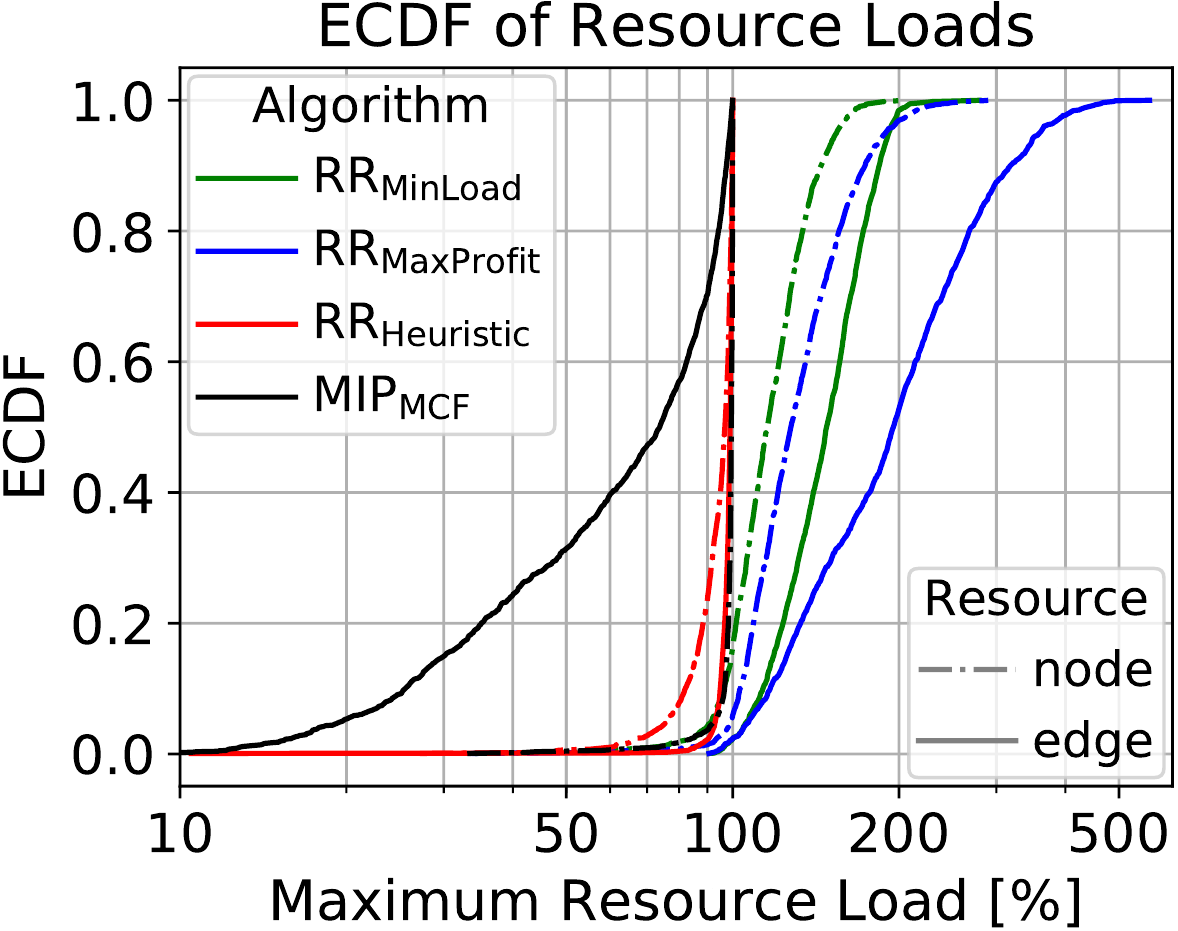}
\hfill
\includegraphics[height=.142\textheight]{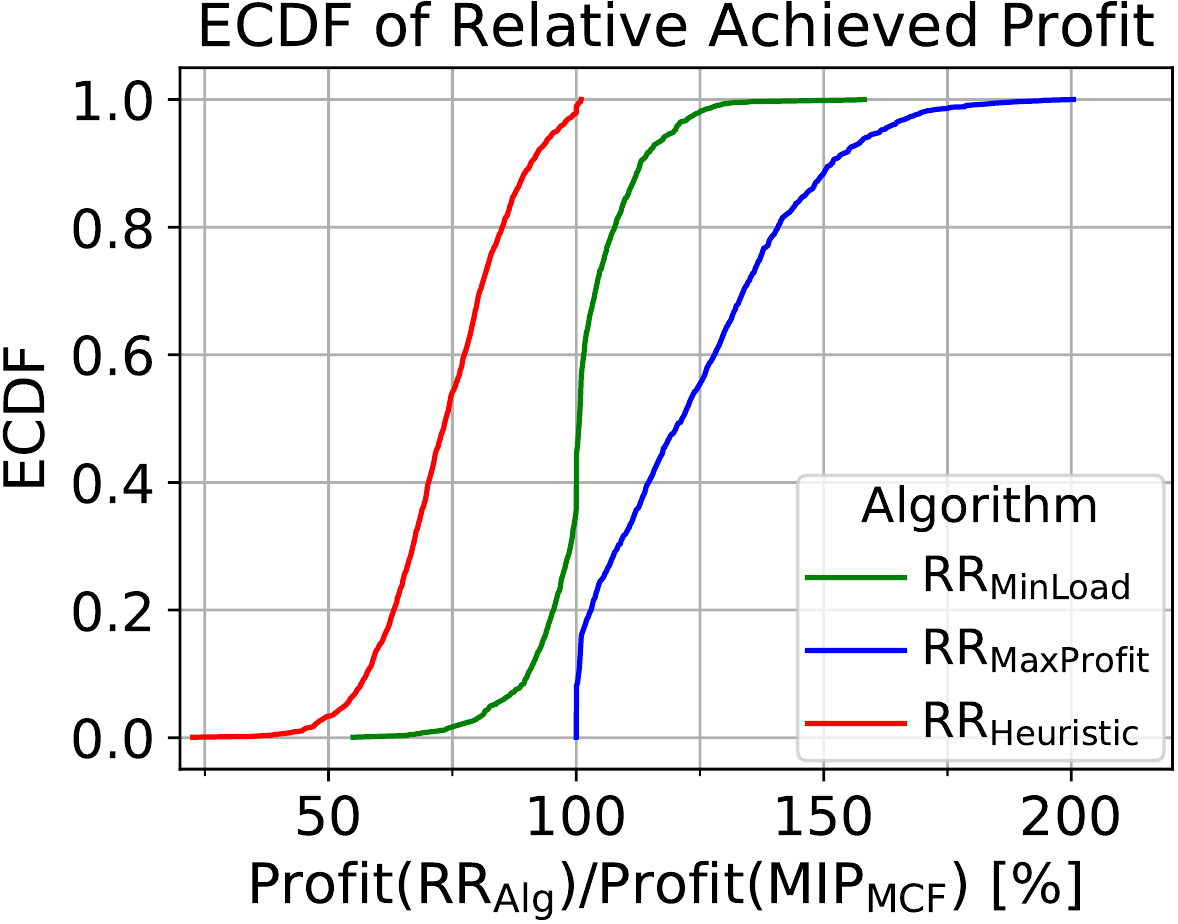}
\caption{Comparison of the different algorithms in terms of (maximal) resource usage and the relative achieved profit for all 1,500 results.\newline
Left: ECDF of maximal node and edge resource loads for all algorithms. For the randomized rounding algorithms, the maximum edge load is in general higher than the maximum edge load. Note the logarithmic x-axis. \newline
Right: ECDF of the profit achieved by the randomized rounding algorithms compared to the profit achieved by the baseline solution.}
\label{fig:ecdfs-load-and-performance}
\end{minipage}
\hfill
\begin{minipage}[b]{0.49\textwidth}
\centering
{
\includegraphics[height=.142\textheight]{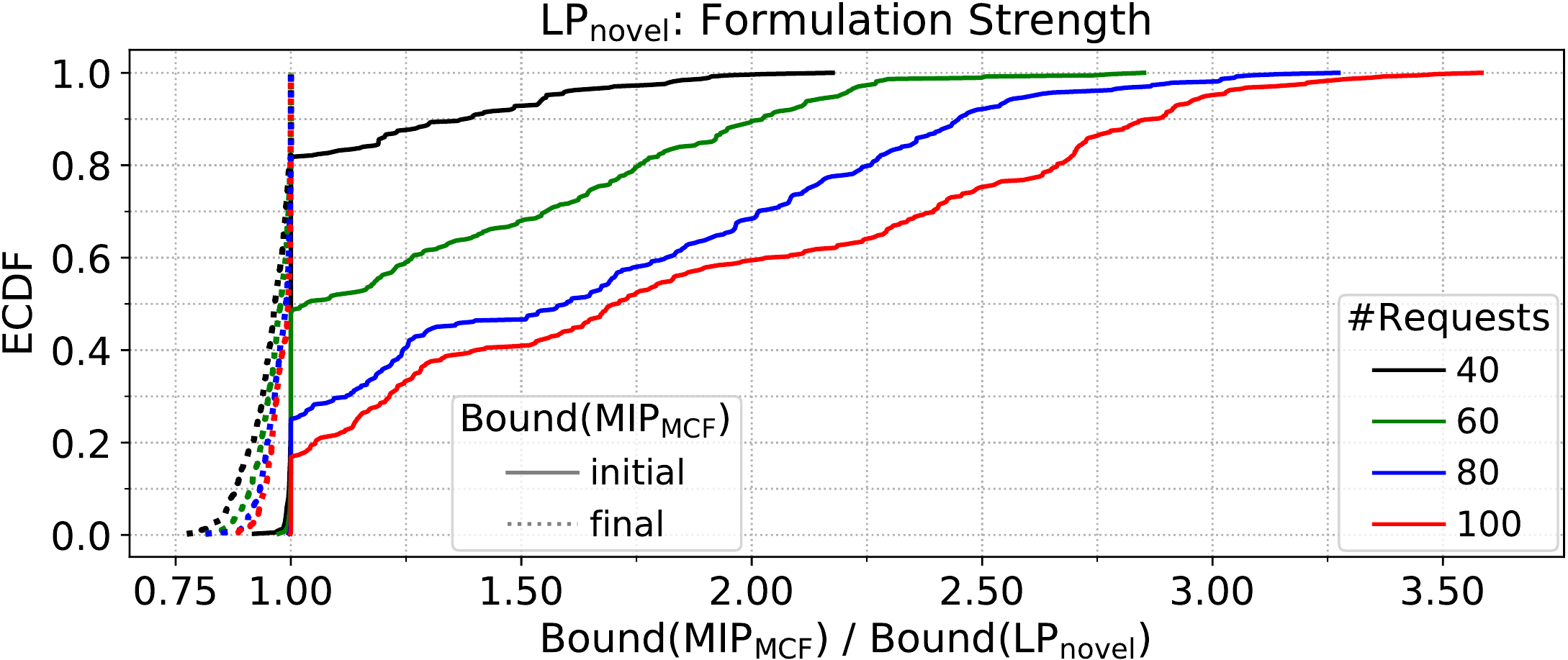}
\caption{Comparison of the bounds on the profit computed by the novel LP~Formulation~\ref{IP:novel} and the classic MCF Formulation~\ref{alg:VNEP-IP-old}. As objective bounds are continuously improved during the solution process of the Mixed-Integer Program $\MIPMCF$, we report on the initial (weakest) bound and the final (best) bound. The initial bound is essentially the objective value of the LP relaxation of Formulation~\ref{alg:VNEP-IP-old}, but might be improved by the solver Gurobi based on the introduction of cutting planes valid only for the integer variant.}
\label{fig:ecdfs-2}

}
\end{minipage}
\end{figure*}

\paragraph{Solving LP Formulation~\ref{IP:novel}} To apply the rounding algorithms presented in Section~\ref{sec:main-section-discussion}, we solve our novel LP Formulation~\ref{IP:novel} by employing again Gurobi 7.5.1, specifically its Barrier algorithm with crossover. Figure~\ref{fig:runtimes} (right) depicts the averaged runtime to solve the LP, including the time to construct the (potentially very large) LP. The latter is not negligible as the formulation contains up to 1,000k variables for some instances. The runtime increases from around $2$ minutes for $|\requests| = 40$ to around $7$ minutes for $|\requests|=100$. The maximally observed runtime in our experiments amounted to roughly 18 minutes.
\paragraph{Vanilla Rounding} With respect to the results of our rounding heuristics, we first discuss the results of our vanilla rounding heuristics $\RRMIN$ and $\RRMAX$. Concretely, we report on the best solution found within 1,000 rounding iterations.
Figure~\ref{fig:randround-results} depicts the respective results as a scatter plot while Figure~\ref{fig:ecdfs-load-and-performance} depicts empirical cumulative distribution functions (ECDF) for both the achieved profit and the maximal resource loads. As can be seen, for $\RRMIN$ the algorithm achieves a profit between 50\% and 140\% compared to the best solution constructed by the MIP, while exceeding resource capacities mostly by 25\% to 125\% of the resource's capacity. For $\RRMAX$, the achieved profit always exceeds the baseline's profit. This is to be expected based on our analysis in Section~\ref{sec:performance-guarantee-obj-admission-control}, as the expected benefit is at least as large as the optimal Mixed-Integer Programming formulation's value. The resource loads mostly lie below 400\% with the maximum coming close to 500\%.

For both selection criteria, the edge resource factor has a distinct impact on the (overall) maximum load. This can be explained as follows. As each request edge may use any of the substrate edges (compared to the restricted node mappings), the chances of `collisions', i.e. multiple request edges being mapped on the same substrate edge, is higher. Additionally, the fact that the number of virtual edges is (always) at least as large as the number of edges (cf. Figure~\ref{fig:request-characteristics}) and that a single request edge may use \emph{multiple} substrate edges, complicates the rounding of edge mappings.
This observation is further substantiated by the ECDF presented in Figure~\ref{fig:ecdfs-2}: the maximal edge load are consistenly larger than the maximal node resource load for all randomized rounding algorithms.

\paragraph{Heuristical Rounding $\RRHEUR$}
The results of the heuristical rounding, which does not exceed resource capacities, are presented in Figure~\ref{fig:randround-results} (heatmaps) and Figure~\ref{fig:ecdfs-2} (ECDFs). Again, 1,000 rounding iterations were considered. While for low edge resource factors, i.e. scarce edge resources, the solutions achieve around 65\% of the profit of the MIP baseline, for larger edge resource factors, the relative performance exceeds 80\%.  Both resource factors have a significant impact and the heuristical rounding procedure works best if neither of the resources is scarce, achieving on average around 97.5\% of the baseline's profit for $\NRF=0.2$ and $\ERF=4.0$. Also, and in contrast to the performance of the baseline $\MIPMCF$, the performance improves when increasing the number of requests. This can be explained as follows. For any combination of resource factors, the virtual resource demands are computed according to Equations~\ref{eq:req-node} and \ref{eq:req-edge} \emph{independently} of how many requests are considered. Hence, when increasing the number of requests, the resource demands of each single request become smaller. Hence, the maximal demand-to-capacity ratio  $\maxDemandX / \Scap(x,y)$ decreases for \emph{all} substrate resources $(x,y) \in \SR$ when increasing the number of requests. Accordingly, the value of $\varepsilon$ (cf. Lemma~\ref{lem:approximation-single-resource} and Corollary~\ref{cor:resource-augmentation-prob}) can be chosen smaller and the expected resource augmentations decrease. Thus, the heuristical rounding algorithm is able to more easily find solutions of high profit \emph{not} augmenting resources.

Overall, the average relative performance with respect to the baseline solutions is 73.8\%, with the minimal one being 22.3\% and only around 5\% of constructed solutions achieving less than 50\% of the baseline's profit. 

\subsection{Comparisong of Formulation Strengths}

Lastly, we empirically study the strength of our novel LP Formulation~\ref{IP:novel} and compare it with the classic MCF Formulation~\ref{alg:VNEP-IP-old}. Concretely, as proven in Theorems~\ref{lem:non-decomposability-integrality-gap} and \ref{lem:non-decomposability-integrality-gap-only-node}, the classic MCF formulation has an unbounded or very large integrality gap in general, i.e. the bound on the profit returned by the solution can be arbitrarily far off the optimal attainable profit. Our novel LP formulation is provably stronger than the old formulation, as it only allows for (fractional) solutions, which can be decomposed into valid mappings. Thus, it will always yield (equal or) better bounds on the attainable profit. 

Figure~\ref{fig:ecdfs-2} presents the experimental comparison of both formulations. In particular, for each of the 1,5000 instances we compare the objective of our novel \emph{Linear Programming} Formulation~\ref{IP:novel} to \emph{two} bounds computed during the solution process of the baseline $\MIPMCF$: the initial bound, i.e. the objective of the \emph{Linear Programming} Formulation~\ref{alg:VNEP-IP-old}, and the final (best) bound computed during the solution process of the (Mixed-)\emph{Integer} Program. 

As can be seen, the initial LP bounds of the classic formulation at times exceeds our formulation's objective by more than 300\%, i.e. the classic formulation `overestimates' the maximal attainable profit by at least a factor of 3. For roughly 60\% of the instances, the novel LP improves the bounds by a factor of 1.5. Clearly, the more requests are considered, the less accurate the classic MCF formulation is. 
Considering the final bound computed during the execution of the \mbox{(Mixed-)\emph{Integer}} Program, we see that these bounds always improve upon upon the novel LP's bound. However, for 80\% of the instances, our LP's bound is only improved by roughly 15\% with the maximal improvement being less than 33\%. 

Concluding, we note that our novel LP formulation is much \emph{stronger} than the classic formulation: it consistently yields significantly better bounds in practice compared to the classic LP formulation and comes close to the bounds obtained by solving the (Mixed-)\emph{Integer} Program for upto 3 hours.

\section{Conclusion}
This paper has initiated the study of approximation algorithms for the Virtual Network Embedding Problem supporting arbitrary substrate graphs and supporting \emph{cyclic} request graphs, specifically cactus request graphs. To obtain the approximation result, we have derived a novel strong LP formulation. Our computational evaluation shows the practical significance of our work: obtained solutions achieve (on average) around 74\% of the baseline's profit while \emph{not augmenting capacities}.

We note that the developed approximation framework -- including the proposed rounding heuristics -- is independent of how LP solutions are computed and decomposed. In particular, while the LP formulation presented in this paper is only applicable for cactus request graphs, our formulation can be generalized to arbitrary request graphs~\cite{rostSchmidFPTApproximationsTechReport}. 

\section*{Acknowledgements} This work was partially supported by Aalborg University's PreLytics project as well as by the German BMBF Software Campus
grant 01IS1205. 

We thank Elias D\"ohne, Alexander Elvers, and Tom Koch for their significant contribution to our  implementation~\cite{github-evaluation}.

{
\balance


}

\end{document}